%% file: Cohen_Zhao_AHT_TIT.tex
\documentclass[11pt,onecolumn,draftclsnofoot]{IEEEtran}
\usepackage{graphicx}
\usepackage{epsfig}
\usepackage{subfigure}
\usepackage{amssymb}
\usepackage{amsbsy}
\usepackage{amsmath}
\usepackage{cite}
\usepackage{stfloats}
\usepackage{url}

\usepackage{color}
\usepackage{placeins}
\usepackage{float}
\usepackage{tabularx,colortbl}
\usepackage{times,amsmath,epsfig}
\usepackage{xspace,latexsym,syntonly}
\usepackage{amssymb}
\usepackage{amsfonts}
\usepackage{textcomp}
\usepackage{subfigure}
\usepackage{amsbsy}
\usepackage{cite}

\input{macros}

\begin{document}\title{Active Hypothesis Testing for Quickest \\Anomaly Detection}
\author{Kobi Cohen and Qing Zhao
\thanks{The authors are with the Department of Electrical and Computer Engineering, University of California, Davis. Email: $\left\{\mbox{yscohen, qzhao}\right\}$@ucdavis.edu}
\thanks{This work was supported by Army Research Lab under Grant W911NF1120086 and by National Science Foundation under Grants CCF-1320065 and CNS-1321115.}
\thanks{Part of this work was presented at the Information Theory and Applications (ITA) Workshop, San Diego, California, USA, Feb. 2014.}
}
\date{}
\maketitle
%
\begin{abstract}
\label{sec:abstract}
The problem of quickest detection of a single anomalous process among a finite number $M$ of processes is considered. At each time, a subset of the processes can be observed, and the observations from each chosen process follow two different distributions, depending on whether the process is normal or abnormal. The objective is a sequential search strategy that minimizes the expected detection time subject to an error probability constraint. This problem can be considered as a special case of active hypothesis testing first considered by Chernoff in 1959 where a randomized strategy, referred to as the Chernoff test, was proposed and shown to be asymptotically (as the error probability approaches zero) optimal. For the special case considered in this paper, we show that a simple deterministic test achieves asymptotic optimality and offers better performance in the finite regime.
We further extend the problem to the case where multiple anomalous processes are present. In particular, we examine the case where only an upper bound on the number of anomalous processes is known.
\end{abstract}
%
\def\keywords{\vspace{.5em}
{\bfseries\textit{Index Terms}---\,\relax%
}}
\def\endkeywords{\par}
\keywords
Sequential detection, anomaly detection, dynamic search, active hypothesis testing, controlled sensing.
\section{Introduction}
\label{sec:intro}

We consider the problem of detecting a single anomalous process among $M$ processes. Borrowing terminologies from target search, we refer to these processes as cells and the anomalous process as the target which can locate in any of the $M$ cells.
The decision maker is allowed to search for the target over $K$ cells at a time ($1\leq K\leq M$). The observations from searching a cell are i.i.d. realizations drawn from two different distributions $f$ and $g$, depending on whether the target is absent or present. The objective is a sequential search strategy that dynamically determines which cells to search at each time and when to terminate the search so that the expected detection time is minimized under a constraint on the probability of declaring a wrong location of the target.

The problem under study applies to intrusion detection in cyber-systems when an intrusion to a subnet has been detected and the objective is to locate the abnormal component in the subnet (since the probability of each component being compromised is small, with high probability, there is only one abnormal component). It also finds applications in target search, fraud detection, and spectrum scanning in cognitive radio networks.

\subsection{A Case of Active Hypothesis Testing}
\label{ssec:AHT}

The above problem is a special case of the sequential experiment design problem first studied by Chernoff in 1959~\cite{Chernoff_1959_Sequential}. Compared with the classic sequential hypothesis testing pioneered by Wald~\cite{Wald_1947_Sequential} where the observation model under each hypothesis is predetermined, the sequential design of experiments has a control aspect that allows the decision maker to choose the experiment to be conducted at each time. Different experiments generate observations from different distributions under each hypothesis. Intuitively, as more observations are gathered, the decision maker becomes more certain about the true hypothesis, which in turn leads to better choices of experiments. Chernoff focused on the case of binary hypotheses and showed that a \emph{randomized} strategy, referred to as the Chernoff test, is asymptotically optimal as the maximum error probability diminishes. Specifically, the Chernoff test chooses the current experiment based on a distribution that depends on past actions and observations.
Variations and extensions of the problem and the Chernoff test were studied in~\cite{Bessler_1960_Theory, Nitinawarat_2012_Controlled, Nitinawarat_2013_Controlled, Nitinawarat_2013_Controlled_cost, Naghshvar_2013_Active, Naghshvar_2013_Sequentiality}, where the problem was referred to as controlled sensing for hypothesis testing in~\cite{Nitinawarat_2012_Controlled, Nitinawarat_2013_Controlled, Nitinawarat_2013_Controlled_cost} and active hypothesis testing in~\cite{Naghshvar_2013_Active, Naghshvar_2013_Sequentiality} (see a more detailed discussion in Section~\ref{ssec:related}).

It is not difficult to see that the quickest anomaly detection problem considered in this paper is a special case of the active hypothesis testing problem considered in~\cite{Chernoff_1959_Sequential, Bessler_1960_Theory, Nitinawarat_2012_Controlled, Nitinawarat_2013_Controlled, Naghshvar_2013_Active, Naghshvar_2013_Sequentiality}. In particular, under each hypothesis that the target is located in a particular cell, the distribution (either $f$ or $g$) of the next observation depends on the cell chosen to be searched. The Chernoff test and its variations proposed in~\cite{Bessler_1960_Theory, Nitinawarat_2012_Controlled, Nitinawarat_2013_Controlled, Naghshvar_2013_Active, Naghshvar_2013_Sequentiality} thus directly apply to our problem. However, in contrast to the randomized nature of the Chernoff test and its variations, we show in this paper that a simple \emph{deterministic} test achieves asymptotic optimality and offers better performance in the finite regime.

\subsection{Main Results}
\label{ssec:main}

Similar to~\cite{Chernoff_1959_Sequential, Bessler_1960_Theory, Nitinawarat_2012_Controlled, Nitinawarat_2013_Controlled, Naghshvar_2013_Active}, we focus on asymptotically optimal policies in terms of minimizing the detection time as the error probability approaches zero. The asymptotic optimality of the Chernoff test as shown in~\cite{Chernoff_1959_Sequential} requires that under any experiment, any pair of hypotheses are distinguishable (i.e., has positive Kullback-Liebler (KL) divergence). This assumption does not hold in the anomaly detection problem considered in this paper. For instance, under the experiment of searching the $i^{th}$ cell, the hypotheses of the target being in the $j^{th}$ ($j\neq i$) and the $k^{th}$ ($k\neq i$) cells yield the same observation distribution $f$. Nevertheless, we show in Theorem~\ref{th:optimality_Chernoff} that the Chernoff test preserves its asymptotic optimality for the problem at hand even without this positivity assumption on all KL divergences. As a result, it serves as a bench mark for comparison.

The Chernoff test, when applied directly to the anomaly detection problem, leads to a randomized cell selection rule: the cells to be searched at the current time are drawn randomly according to a distribution determined by past observations and actions. The main result of this paper is to show that a simple deterministic policy offers the same asymptotic optimality yet with significant performance gain in the finite regime and considerable reduction in implementation complexity. Specifically, under the proposed policy, the selection rule $\phi(n)$ indicating which $K$ cells should be searched at time~$n$ is given by:
\begin{center}
$\bea{l}
\displaystyle \phi(n)=
\begin{cases} \left(m^{(1)}(n), m^{(2)}(n), ..., m^{(K)}(n)\right) \;\;, \vspace{0.1cm}\\ \hspace{2.0cm}
                                                \mbox{if \;$D(g||f)\geq\frac{D(f||g)}{(M-1)}$ or $K=M$}   \vspace{0.3cm}\\
              \left(m^{(2)}(n), m^{(3)}(n), ..., m^{(K+1)}(n)\right) \;\;,\vspace{0.1cm}\\ \hspace{2.0cm} \mbox{if \;$D(g||f)<\frac{D(f||g)}{(M-1)}$ and $K<M$}
\end{cases}
\ena$
\end{center}
where $m^{(i)}(n)$ denotes the cell index with the $i^{th}$ highest sum of log-likelihood ratio (LLR) collected from this cell up to time~$n$, and $D(\cdot||\cdot)$ is the KL divergence between two distributions. Since $D(g||f)$ is the key quantity in the cell selection rule, we refer to the proposed deterministic policy as the DGF policy.

This deterministic selection rule is intuitively satisfying. Consider, for example, $K=1$. In this case, the DGF policy selects, at each time, either the cell with the largest sum LLRs or the cell with the second largest sum LLRs, depending on the order of $D(g||f)$ and $D(f||g)/(M-1)$. The intuition behind this selection rule is that $D(g||f)$ and $D(f||g)/(M-1)$ determine, respectively, the rates at which the state of the cell with the target and the states of the $M-1$ cells without the target can be accurately inferred. Based on the order of these two rates, the DGF policy aims at identifying either the cell with the target or those $M-1$ cells without the target. The selection rule is thus clear by noticing that searching the cell with the second largest sum LLRs will lead to sufficient exploration of all $M-1$ cells without the target since the less explored cells tend to have higher sum LLRs among these $M-1$ cells. A more detailed discussion of the DGF policy and a rigorous proof of its asymptotic optimality are given in Section~\ref{sec:DGF}.

We then extend the problem to the case where multiple anomalous processes are present. In particular, we examine the case where only an upper bound on the number of anomalous processes is known. Interestingly, we show that the Chernoff test may not be practically appealing under the latter setting. We thus consider a modified Bayes risk that better captures the design objective of practical systems and develop a deterministic policy that is again asymptotically optimal.

\subsection{Related Work}
\label{ssec:related}

Chernoff's pioneer work on active hypothesis testing focuses on sequential binary composite hypothesis testing~\cite{Chernoff_1959_Sequential}. The extension to M-ary hypothesis was given by Bessler in~\cite{Bessler_1960_Theory}. In~\cite{Nitinawarat_2013_Controlled}, Nitinawarat et al. considered M-ary active hypothesis testing in both fixed sample size and sequential settings. Under the sequential setting, they developed a modified Chernoff test that is asymptotically optimal without the positivity assumption on all KL divergences as required in~\cite{Chernoff_1959_Sequential, Bessler_1960_Theory}. Furthermore, they examined the asymptotic optimality of the Chernoff test under constraints on decision risks, a stronger condition than the error probability, and developed a modified Chernoff test to meet hard constraints on the decision risks. In \cite{Nitinawarat_2013_Controlled_cost}, a more general model of Markovian Observations and non-uniform
control cost was considered. In~\cite{Naghshvar_2013_Active}, in addition to the asymptotic optimality adopted by Chernoff in~\cite{Chernoff_1959_Sequential}, Naghshvar and Javidi examined active sequential hypothesis testing under the notion of non-zero information acquisition rate by letting the number of hypotheses approach infinity and under a stronger notion of asymptotic optimality. They further studied in~\cite{Naghshvar_2013_Sequentiality} the roles of sequentiality and adaptivity in active hypothesis testing by characterizing the gain of sequential tests over fixed sample size tests and the gain of closed-loop policies over open-loop policies.

Target search or target whereabout problems have been widely studied under various scenarios. Results under the sequential setting can be found in~\cite{Zigangirov_1966_Problem, Klimko_1975_Optimal, Dragalin_1996_Simple, Stone_1971_Optimal}, all assuming single process observations (i.e., $K=1$). Specifically, optimal policies were derived in~\cite{Zigangirov_1966_Problem, Klimko_1975_Optimal, Dragalin_1996_Simple} for the problem of quickest search over Weiner processes. In~\cite{Stone_1971_Optimal}, an optimal search strategy was established under the constraint that switching to a new process is allowed only when the state of the currently probed process is declared. Optimal policies under general distributions or with general multi-process probing strategies remain an open question. In this paper we address these questions under the asymptotic regime as the error probability approaches zero.
Target search with a fixed sample size was considered in~\cite{Tognetti_1968_An, Kadane_1971_Optimal, Zhai_2013_Dynamic, Castanon_1995_Optimal}. In~\cite{Tognetti_1968_An, Kadane_1971_Optimal, Zhai_2013_Dynamic}, searching in a specific location provides a binary-valued measurement regarding the presence or absence of the target. Similar to this paper, Castanon considered in~\cite{Castanon_1995_Optimal} continuous observations: the observations from a location without the target and with the target have distributions $f$ and $g$, respectively. Different from this paper where we consider sequential settings and obtain an asymptotically optimal policy that applies to general distributions,~\cite{Castanon_1995_Optimal} focused on the fixed sample size setting and required a symmetry assumption on the distributions (specifically, $f(x)=g(b-x)$ for some $b$) for the optimality of the proposed policy. The problem of universal outlier hypothesis testing was studied in~\cite{Li_2013_Universal}. Under this setting, a vector of observations containing coordinates with an outlier distribution is observed at each given time. The goal is to detect the coordinates with the outlier distribution based on a sequence of $n$ i.i.d. vectors of observations.

Another set of related work is concerned with sequential detection over multiple independent processes~\cite{Zhao_2010_Quickest, Li_2009_Restless, Caromi_2013_Fast, Cohen_2013_Optimal_GlobalSIP, Cohen_2014_Optimal, Cohen_2014_Asymptotically, Lai_2011_Quickest, Malloy_2012_Sequential, Malloy_2012_Quickest, Tajer_2013_Quick, Hadjiliadis_2009_One, Raghavan_2010_Quickest, Bayraktar_2013_Byzantine, Zhang_2014_Quickest}. In particular, in~\cite{Lai_2011_Quickest}, the problem of identifying the first abnormal sequence among an infinite number of i.i.d. sequences was considered. An optimal cumulative sum (CUSUM) test was established under this setting. Further studies on this model can be found in~\cite{Malloy_2012_Sequential, Malloy_2012_Quickest, Tajer_2013_Quick}. While the objective of finding rare events or a single target considered in~\cite{Lai_2011_Quickest, Malloy_2012_Sequential, Malloy_2012_Quickest, Tajer_2013_Quick} is similar to that of this paper, the main difference is that in~\cite{Lai_2011_Quickest, Malloy_2012_Sequential, Malloy_2012_Quickest, Tajer_2013_Quick} the search is done over an infinite number of i.i.d processes, where the state of each process (normal or abnormal) is independent of other processes. Under this independence assumption, the structure of the solution is to perform an independent sequential test without memory for each process. At each time when the decision maker decides to switch to a different process, the new process is chosen arbitrarily, and a sequential test starts afresh. In this paper, however, the number of the processes is finite and the number of the abnormal ones is known (or an upper bound is known). As a result, the process states are correlated. Under this model, the selection rule that governs which process to observe at each time is crucial in minimizing the detection delay, whereas in~\cite{Lai_2011_Quickest, Malloy_2012_Sequential, Malloy_2012_Quickest, Tajer_2013_Quick} the order at which the processes are observed is irrelevant. Furthermore, in our model, the sequential tests for the processes have memory. When a process is revisited, all the observations obtained during previous visits are taken into consideration in decision making.

Another related problem considered recently deals with detecting the first disorder of a system involving multiple processes~\cite{Hadjiliadis_2009_One, Raghavan_2010_Quickest, Bayraktar_2013_Byzantine, Zhang_2014_Quickest}. In this problem, multiple sensors take observations sequentially from the environment and communicate with a fusion center, which determines whether there is a change in the statistical behavior of the observations. The asymptotic optimality of the multi-chart CUSUMs in detecting the first change-point was studied as the mean time between false alarms approaches to infinity. In~\cite{Hadjiliadis_2009_One}, asymptotic optimality was shown under one-shot schemes, in which the sensors communicate with the fusion center only when they signal an alarm. A Bayesian version of this problem was considered in~\cite{Raghavan_2010_Quickest} under the assumption that the fusion center has perfect information about the observations and a priori knowledge of the statistics of the change process. In~\cite{Bayraktar_2013_Byzantine}, the problem was examined for the case where an unknown subset of sensors are compromised and a fully distributed low complexity detection scheme was proposed to mitigate the performance degradation and recover the log scaling. In~\cite{Zhang_2014_Quickest}, asymptotic optimality of the multi-chart CUSUMs was shown under a coupled system, where observations in one sensor can affect the observations in another. In this paper, however, the goal is to detect the abnormal processes (and not a change point), where the process states are fixed during the detection process.

\subsection{Organization}
\label{ssec:organization}

In Section~\ref{sec:problem} we describe the system model and problem formulation. In Section~\ref{sec:DGF} we propose the deterministic DGF policy and establish its asymptotic optimality. We also provide a comparison of DGF with the randomized Chernoff test. In Section~\ref{sec:extension} we extend the problem to the case where multiple anomalous processes are present and consider both cases of known and unknown number of anomalous processes. In Section~\ref{sec:simulation} we provide numerical examples to illustrate the performance of the proposed policy as compared with the Chernoff test. Section~\ref{sec:conclusion} concludes the paper.

\section{System Model and Problem Formulation}
\label{sec:problem}

\subsection{System Model}
\label{ssec:system}

Consider the following anomaly detection problem. A decision maker is required to detect the location of a single anomalous object (referred as a target) located in one of $M$ cells. If the target is in cell $m$, we say that hypothesis $H_m$ is true. The \emph{a priori} probability that $H_m$ is true is denoted by $\pi_m$, where $\sum_{m=1}^{M}{\pi_m}=1$. To avoid trivial solutions, it is assumed that $0<\pi_m<1$ for all $m$.

At each time, only $K$ ($1\leq K\leq M$) cells can be observed. When cell $m$ is observed at time~$n$, an observation $y_m(n)$ is drawn independently from a distribution in a one-at-a-time manner. If hypothesis $m$ is false, $y_m(n)$ follows distribution $f(y)$; if hypothesis $m$ is true, $y_m(n)$ follows distribution $g(y)$. Let $\mathbf{P}_m$ be the probability measure under hypothesis $H_m$ and $\E_m$ the operator of expectation with respect to the measure $\mathbf{P}_m$.

We define the stopping rule $\tau$ as the time when the decision maker finalizes the search by declaring the location of the target.
Let $\delta\in\left\{1, 2, ..., M\right\}$ be a decision rule, where $\delta=m$ if the decision maker declares that $H_m$ is true.
Let $\phi(n)\in\left\{1, 2, ..., M\right\}^K$ be a selection rule indicating which $K$ cells are chosen to be observed at time~$n$. The time series vector of selection rules is denoted by $\boldsymbol\phi=(\phi(n), n=1, 2, ...)$.
Let $\mathbf{y}_{\phi(n)}(n)$ be the vector of observations obtain from cells $\phi(n)$ at time~$n$ and $\mathbf{y}(n)=\left\{\phi(t), \mathbf{y}_{\phi(t)}(t)\right\}_{t=1}^n$ be the set of all cell selections and observations up to time~$n$. A deterministic selection rule $\phi(n)$ at time~$n$ is a mapping from $\mathbf{y}(n-1)$ to $\left\{1, 2, ..., M\right\}^K$. A randomized selection rule $\phi(n)$ is a mapping from $\mathbf{y}(n-1)$ to probability mass functions over $\left\{1, 2, ..., M\right\}^K$. \vspace{0.2cm}
\begin{definition}
An admissible strategy $\Gamma$ for the sequential anomaly detection problem is given by the tuple $\Gamma=(\tau, \delta, \boldsymbol\phi)$.
\end{definition}

\subsection{Objective}
\label{ssec:objective}

Let $P_e(\Gamma)=\sum_{m=1}^{M}{\pi_m\alpha_m(\Gamma)}$ be the probability of error under strategy $\Gamma$, where $\alpha_m(\Gamma)=\mathbf{P}_m(\delta\neq m|\Gamma)$ is the probability of declaring $\delta\neq m$ when $H_m$ is true. Let $\mathbf{E}(\tau|\Gamma)=\sum_{m=1}^{M}{\pi_m\E_m(\tau|\Gamma)}$ be the average detection delay under $\Gamma$.

We adopt a Bayesian approach as in~\cite{Chernoff_1959_Sequential, Nitinawarat_2012_Controlled} by assigning a cost of $c$ for each observation and a loss of $1$ for a wrong declaration. The Bayes risk under strategy $\Gamma$ when hypothesis $H_m$ is true is given by:
\beq
\label{eq:Bayes_risk_m}
\displaystyle R_m(\Gamma)\triangleq\alpha_m(\Gamma)+c\E_m(\tau|\Gamma) \;.
\eeq
Note that $c$ represents the ratio of the sampling cost to the cost of wrong detections.\\
The average Bayes risk is given by:
\beq
\label{eq:Bayes_risk}
\displaystyle R(\Gamma)=\sum_{m=1}^{M}\pi_m R_m(\Gamma)=P_e(\Gamma)+c\mathbf{E}(\tau|\Gamma) \;.
\eeq

The objective is to find a strategy $\Gamma$ that minimizes the Bayes risk $R(\Gamma)$:
\beq\label{eq:Bayes_formulation1}
\displaystyle\inf_{\Gamma} \;\; R(\Gamma) \;.
\eeq

\subsection{Notations}
\label{ssec:notations}

Let $\mathbf{1}_m(n)$ be the indicator function, where $\mathbf{1}_m(n)=1$ if cell $m$ is observed at time~$n$, and $\mathbf{1}_m(n)=0$ otherwise.
Let
\beq
\label{eq:sum_LLR}
\displaystyle \ell_m(n)\triangleq\log \frac{g(y_m(n))}{f(y_m(n))} \;,
\eeq
and
\beq
\label{eq:sum_LLR}
\displaystyle S_m(n)\triangleq\sum_{t=1}^{n}{\ell_m(t)\mathbf{1}_m(t)}
\eeq
be the log-likelihood ratio (LLR) and the observed sum LLRs of cell $m$ at time~$n$, respectively. We then define $m^{(i)}(n)$ as the index of the cell with the $i^{th}$ highest observed sum LLRs at time~$n$.
Let
\beq
\label{eq:Delta_S}
\Delta S(n)\triangleq S_{m^{(1)}(n)}(n)-S_{m^{(2)}(n)}(n)
\eeq
denote the difference between the highest and the second highest observed sum LLRs at time~$n$.\\
Finally, we define
\beq
\bea{l}
\displaystyle I^*(M,K)\triangleq \vspace{0.2cm}\\ \hspace{0.5cm}
\begin{cases} \displaystyle D(g||f)+D(f||g)\;\;,\;\;\mbox{if \;$K=M$\;,} \vspace{0.3cm}\\
                                 \displaystyle\max\left[\frac{K D(f||g)}{M-1}, D(g||f)+\frac{(K-1)D(f||g)}{M-1}\right] \;\;,
                                                                    \vspace{0.2cm}\\ \hspace{6cm}
                                                             \mbox{if $K<M$\;.} \vspace{0.2cm}
        \end{cases}
\ena
\eeq
In subsequent sections we show that $I^*(M,K)$ plays the role of the rate function, which determines the asymptotically optimal performance of the test. Increasing $I^*(M,K)$ decreases the asymptotic lower bound on the Bayes risk. It is intuitive that $I^*(M,K)$ increases with the observation capability $K$ and decreases with the hypothesis size $M$.

\section{The Deterministic DGF Policy}
\label{sec:DGF}

In this section we propose a deterministic policy, referred to as the DGF policy, to solve (\ref{eq:Bayes_formulation1}). Theorem~\ref{th:optimality_policy1} shows that the DGF policy is asymptotically optimal in terms of minimizing the Bayes risk (\ref{eq:Bayes_risk}) as $c\rightarrow 0$.

\subsection{The DGF Policy}

At each time~$n$, the selection rule $\phi(n)$ of the DGF policy chooses cells according to the order of their sum LLRs. Specifically, based on the relative order of $D(g||f)$ and $D(f||g)/(M-1)$, either the $K$ cells with the top $K$ highest sum LLRs or those with the second to the $(K+1)^{th}$ highest sum LLRs are chosen, i.e.,\footnote{Cells with the same sum LLRs can be ordered arbitrarily.}
\beq
\label{eq:selection_policy1}
\displaystyle \phi(n)=
\begin{cases} \left(m^{(1)}(n), m^{(2)}(n), ..., m^{(K)}(n)\right) \;\;, \vspace{0.1cm}\\ \hspace{2.0cm}
                                                \mbox{if \;$D(g||f)\geq\frac{D(f||g)}{(M-1)}$ or $K=M$}   \vspace{0.3cm}\\
              \left(m^{(2)}(n), m^{(3)}(n), ..., m^{(K+1)}(n)\right) \;\;,\vspace{0.1cm}\\ \hspace{2.0cm} \mbox{if \;$D(g||f)<\frac{D(f||g)}{(M-1)}$ and $K<M$}
\end{cases} \;,
\eeq
The stopping rule and decision rule under the DGF policy are given by:
\beq
\label{eq:stopping_policy1}
\bea{l}
\displaystyle \tau= \inf \left\{n \; : \; \Delta S(n)\geq -\log c\right\}\;,
\ena
\eeq
and
\beq
\label{eq:decision_policy1}
\displaystyle\delta= m^{(1)}(\tau) \;.
\eeq

The deterministic selection rule of the DGF policy can be intuitively explained as follows. Consider the case where $K=1$. If cell $m^{(1)}(n)$ is selected at each given time~$n$, the asymptotic detection time approaches $-\log c/D(g||f)$ since the cell with the target (say $m$) is observed at each given time with high probability (in the asymptotic regime) and the test is finalized once sufficient information is gathered from this cell (for a detailed asymptotic analysis see Appendix~\ref{app:proof_policy1}). In this case, $D(g||f)$ determines the asymptotically optimal performance of the test since $\E_m(\ell_m)=D(g||f)$.
On the other hand, if cell $m^{(2)}(n)$ is selected at each given time~$n$, the asymptotic detection time approaches $-(M-1)\log c/D(f||g)$ since one of the $M-1$ cells without the target is observed at each given time with high probability and the test is finalized once sufficient information is gathered from all these cells. Since $\E_m(\ell_j)=-D(f||g)$ for all $j\neq m$, the asymptotically optimal performance of the test is determined by $D(f||g)/(M-1)$. Therefore, the selection rule selects the strategy that minimizes the asymptotic detection time according to $\max\left[D(g||f),D(f||g)/(M-1)\right]$. When $K>1$, the rates at which the state of cell $m$ and the states of the rest $M-1$ cells can be accurately inferred are given by $D(g||f)+\frac{(K-1)D(f||g)}{M-1}$ and $\frac{K D(f||g)}{M-1}$, respectively. Since $D(g||f)>D(f||g)/(M-1)$ is equivalent to $D(g||f)+\frac{(K-1)D(f||g)}{M-1}>\frac{K D(f||g)}{M-1}$, the selection rule of DGF is thus clear.

\subsection{Performance Analysis}
\label{sec:performance}

The following main theorem shows that the DGF policy is asymptotically optimal in terms of minimizing the Bayes risk as $c$ approaches zero:\vspace{0.2cm}

\begin{theorem}[asymptotic optimality of the DGF policy]
\label{th:optimality_policy1}
Let $R^*$ and $R(\Gamma)$ be the Bayes risks under the DGF policy and any other policy $\Gamma$, respectively. Then\footnote{The notation $f\sim g$ as $c\rightarrow 0$ refers to $\lim_{c\rightarrow 0}f/g=1$.},
\beq
\label{eq:asymptotic_performance}
R^* \;\sim\; \frac{-c\log c}{I^*(M,K)}\;\sim\;\inf_{\Gamma}\;{R(\Gamma)} \;\;\; \mbox{as} \;\;\; c\rightarrow 0 \;.
\vspace{0.2cm}
\eeq
\end{theorem}
\begin{proof}
For a detailed proof see Appendix~\ref{app:proof_policy1}. We provide here a sketch of the proof. In App.~\ref{ssec:lower_bound}, we show that $\frac{-c\log c}{I^*(M,K)}$ is an asymptotic lower bound on the achievable Bayes risk.
Then, we show in App.~\ref{ssec:asymptotic_policy1} that the Bayes risk $R^*$ under the DGF policy approaches the asymptotic lower bound as $c\rightarrow 0$. Specifically, the asymptotic behavior of $R^*$ is established based on Lemma~\ref{lemma:expected_time_policy1} showing that the asymptotic expected detection time approaches $\frac{-\log c}{I^*(M,K)}$, while the error probability is $O(c)$ following Lemma~\ref{lemma:error_policy1}.

The basic idea in establishing the asymptotic expected detection time under DGF in Lemma~\ref{lemma:expected_time_policy1} is to upper bound the stopping time $\tau$ of DGF by analyzing three last passage times (given in Lemmas~\ref{lemma:tau_1_policy1},~\ref{lemma:tau_2_policy1} and~\ref{lemma:tau_3_policy1}). Specifically, if the stopping rule is disregarded and sampling is continued indefinitely, then three last passage times can be defined: $\tau_1, \tau_2, \tau_3$, where, roughly speaking, $\tau_1$ is the time when the sum LLRs of the true cell (say $m$) is the highest among all the cells for all $n\geq\tau_1$; $\tau_2$ is the time when sufficient information for distinguishing hypothesis $m$ from at least one false hypothesis has been gathered; $\tau_3$ is the time when sufficient information for distinguishing hypothesis $m$ from all false hypotheses has been gathered. It should be noted that $\tau_1, \tau_2, \tau_3$ are not stopping times and the decision maker does not know whether they have arrived (since the true cell is unknown and also $\tau_1, \tau_2, \tau_3$ depend on the future by definition). However, by the definition of $\tau_3$ (see Definition $7$ in Appendix~\ref{app:proof_policy1} for details) the actual stopping time $\tau$ under DGF is upper bounded by $\tau_3$ (i.e., the decision maker does know that for all $n<\tau$, $\tau_3$ surely has not arrived). As a result, $\E(\tau_3)$ is an upper bound of $\E(\tau)$.

To show the asymptotic behavior of $\E(\tau_3)$, define $n_2=\tau_2-\tau_1$ and $n_3=\tau_3-\tau_2$. Thus, $\tau_3=\tau_1+n_2+n_3$. Lemma~\ref{lemma:tau_2_policy1} shows that $\E(n_2)\sim-\log c/I^*(M,K)$ as $c\rightarrow 0$.
Lemma~\ref{lemma:tau_1_policy1} shows that $\E(\tau_1)/\E(n_2)\rightarrow 0$, i.e., $\tau_1$ does not affect the asymptotic detection time. Note that differing from~\cite{Nitinawarat_2013_Controlled}, where only polynomial decay of $\mathbf{P}_m(\tau_1>n)$ was shown under the extended Chernoff test developed to handle indistinguishable hypotheses under some actions, Lemma~\ref{lemma:tau_1_policy1} shows exponential decay of $\mathbf{P}_m(\tau_1>n)$ under DGF. Lemma~\ref{lemma:tau_3_policy1} shows that $\E(n_3)/\E(n_2)\rightarrow 0$.
Combining Lemmas~\ref{lemma:tau_1_policy1},~\ref{lemma:tau_2_policy1} and~\ref{lemma:tau_3_policy1}, we can conclude that $\E(\tau_3)\sim -\log c/ I^*(M,K)$. Since the error probability is $O(c)$ following Lemma~\ref{lemma:error_policy1}, the proof thus completes by noticing that the upper bound on $c\E(\tau)+P_e$ coincides with the lower bound on the achievable Bayes risk.
\end{proof}

\subsection{Comparison with the Chernoff Test}
\label{ssec:comparison}

Next, we analyze the classic randomized Chernoff test proposed in~\cite{Chernoff_1959_Sequential} when it is applied to the anomaly detection problem. We then compare the performance of the proposed DGF policy with the Chernoff test. \vspace{0.2cm}

\subsubsection{The Chernoff Test}
\label{sssec:Chernoff}

The Chernoff test has a randomized selection rule. Specifically, let $q=(q_1, ..., q_N)$ be a probability mass function over a set of $N$ available experiments $u=\left\{u_i\right\}_{i=1}^N$ that the decision maker can choose from, where $q_i$ is the probability of choosing experiment $u_i$. For a general M-ary active hypothesis testing problem, the action at time~$n$ under the Chernoff test is drawn from a distribution $q^*(n)=(q^*_1(n), ..., q^*_N(n))$ that depends on the past actions and observations:
\beq
\label{eq:selection_Chernoff}
\displaystyle q^*(n)=\arg\;\max_{q}\;\min_{j\in\mathcal{M}\setminus\left\{\hat{i}(n)\right\}}
\sum_{u_i}q_i D(p_{\hat{i}(n)}^{u_i}||p_j^{u_i})\;,
\eeq
where $\mathcal{M}$ is the set of the $M$ hypotheses, $\hat{i}(n)$ is the ML estimate of the true hypothesis at time~$n$ based on past actions and observations, and $p_j^{u_i}$ is the observation distribution under hypothesis $j$ when action $u_i$ is taken. The stopping rule and decision rule are given in (\ref{eq:stopping_policy1}), (\ref{eq:decision_policy1})

It can be shown that when applied to the anomaly detection problem, the Chernoff test works as follows. When $D(g||f)\geq D(f||g)/(M-1)$, the Chernoff test selects cell $m^{(1)}(n)$ and draws the rest $K-1$ cells randomly with equal probability from the remaining $M-1$ cells. When $D(g||f)<D(f||g)/(M-1)$, all $K$ cells are drawn randomly with equal probability from cells $\{m^{(2)}(n),m^{(3)}(n),\ldots, m^{(M)}(n)\}$ under the Chernoff test.

Even though the positivity assumption on KL divergences as required in the proof of the asymptotic optimality of the Chernoff test given in~\cite{Chernoff_1959_Sequential} no longer holds for the anomaly detection problem, we show in Theorem~\ref{th:optimality_Chernoff} below that the Chernoff test preserves its asymptotic optimality in this case. Note that in~\cite{Nitinawarat_2013_Controlled}, a modified Chernoff test was developed in order to handle indistinguishable hypotheses under some (but not all) actions. The basic idea of the modified test is to replace the action distribution given in (\ref{eq:selection_Chernoff}) with a uniform distribution for a subsequence of time instants that grows at a sublinear rate with time. This subsequence of arbitrary actions are independent of past observations and affects the finite-time performance. In Theorem~\ref{th:optimality_Chernoff} below we show that this modification is unnecessary for the anomaly detection problem.

\begin{theorem}
\label{th:optimality_Chernoff}
Let $R_{CT}$ and $R(\Gamma)$ be the Bayes risks under the Chernoff test and any other policy $\Gamma$, respectively. Then,
\beq
\label{eq:asymptotic_performance_Chenoff}
R_{CT} \;\sim\; \frac{-c\log c}{I^*(M,K)}\;\sim\;\inf_{\Gamma}\;{R(\Gamma)} \;\;\; \mbox{as} \;\;\; c\rightarrow 0 \;.
\vspace{0.2cm}
\eeq
\end{theorem}
\begin{proof}
The proof is given in Appendix~\ref{app:proof_Chernoff} and is based on the argument of~\cite{Nitinawarat_2013_Controlled} and the proof of Theorem~\ref{th:optimality_policy1} given in Appendix~\ref{app:proof_policy1}. \vspace{0.2cm}
\end{proof}

\subsubsection{Comparison}
\label{sssec:comparison}

Although both the Chernoff test and the DGF policy are asymptotically optimal, simulation results demonstrate significant performance gain of DGF over the Chernoff test in the finite regime (see Section~\ref{sec:simulation}). Next, we provide an intuition argument for the better finite-time performance of DGF by drawing an analogy between the anomaly detection problem and the makespan scheduling problem.

Consider the problem of scheduling $M$ jobs over $K$ parallel machines ($K\leq M$). Each job requires a deterministic processing time of $T_p$ time units. The objective is to minimize the makespan which is defined as the completion time of all $M$ jobs. Note that when $K>1$, processing a job continuously until it is completed can be highly suboptimal since a certain number of machines are left idle when there are less than $K$ unfinished jobs. Note also that keeping machines idle during the scheduling process increases the makespan for all $K\geq 1$. The optimal solution to this problem is given by the LRPT (the longest remaining processing time first) scheduler~\cite[Theorem 5.2.7]{Pinedo_2012_Scheduling} that schedules, at any time~$n$, the $K$ jobs with the longest remaining processing time.

The anomaly detection problem can be viewed as a problem of scheduling $M-1$ jobs (each being the detection process of distinguishing one of the $M-1$ false hypotheses from the true hypothesis) over $K$ machines (which is the number of cells that the decision maker can probe simultaneously). Consider first $D(g||f)<D(f||g)/(M-1)$. In this case, DGF probes cells $\left(m^{(2)}(n), m^{(3)}(n), ..., m^{(K+1)}(n)\right)$ at each time, while the Chernoff test selects $K$ cells randomly among the cells $\left(m^{(2)}(n), m^{(3)}(n), ..., m^{M}(n)\right)$. Both tests terminate once $\Delta S(n)\geq-\log c$ occurs.
Assume that hypothesis $H_m$ is true. Roughly speaking, following Lemma~\ref{lemma:error_policy1}, once $\Delta S_{m,j}\triangleq S_m(n)-S_j(n)>-\log c$, the decision maker has sufficient evidence to distinguish false hypothesis $H_j$ from the true hypothesis $H_m$. Except during an asymptotically insignificant initial stage of the detection process, cells $\left(m^{(2)}, ..., m^{(M)}(n)\right)$ are the cells without the target (see Lemma~\ref{lemma:tau_1_policy1} for a detailed analysis on the last passage time $\tau_1$ of cell $m^{(1)}(n)$ being the cell with the target for all $n\geq\tau_1$). In this case, cells
$\left(m^{(2)}, ..., m^{(K+1)}(n)\right)$ as selected by DGF can be viewed as the cells with the longest remaining processing times. The randomized Chenoff test, however, may lead to inefficient exploitation of the probing capacity, as explain above for the makespan scheduling problem.
Furthermore, randomly selecting $K$ cells from $\left(m^{(2)}, ..., m^{(M)}(n)\right)$ may result in probing a cell whose state can already be inferred with sufficient accuracy (i.e., $\Delta S_{m,j}>-\log c$ as detailed in Appendix~\ref{app:proof_policy1}), which can be viewed as scheduling a job that is already completed or equivalently, leaving a machine idle in the makespan problem. Such actions, however, will not occur under DGF.
The argument for the case of $D(g||f)>D(f||g)/(M-1)$ is similar by viewing the problem as scheduling $M-1$ jobs over $K-1$ machines. Note that both DGF and the Chernoff test dedicate one machine for probing cell $m^{(1)}(n)$ since under the condition of $D(g||f)>D(f||g)/(M-1)$, probing the cell with the target is preferred to accelerate the detection process.

\section{Extension to Multiple Anomalous Processes}
\label{sec:extension}

In this section we extend the results reported in previous sections to the case where multiple processes are abnormal. In Section~\ref{ssec:L} we consider the detection of $L$ abnormal processes, where $L$ is known. In Section~\ref{ssec:upper_bound_L} we consider the case where an unknown number $\ell\geq 1$ of abnormal processes are present and only an upper bound $\ell\leq L$ is known.

Throughout this section, we define $\mathcal{M}'$ as the set of all possible combinations of target locations, with cardinality $M'=|\mathcal{M}'|$ (i.e., a set of $M'$ hypotheses, $H_{m'}$, indicating that the locations of all targets are given by the $(m')^{th}$ set in $\mathcal{M}'$) and $\pi_{m'}$ as the \emph{a priori} probability that $H_{m'}$ is true. Here, the decision rule declares a set of target locations (i.e., hypothesis $H_{m'}$) and the error probability under policy $\Gamma$ is defined as $P_e(\Gamma)=\sum_{m'=1}^{M'}{\pi_{m'}\alpha_{m'}(\Gamma)}$, where $\alpha_{m'}(\Gamma)=\mathbf{P}_{m'}(\delta\neq H_{m'}|\Gamma)$ is the probability of declaring $\delta\neq H_{m'}$ when $H_{m'}$ is true.

\subsection{Known Number of Abnormal Processes}
\label{ssec:L}

Consider the case where $L$ abnormal processes are located among the $M$ cells and $L$ is known. In this case, the detection problem involves $M'=\binom{M}{L}$ hypotheses. We show below that a variation of the DGF policy, dubbed the DGF(L) policy, is asymptotically optimal under this setting.

The stopping rule and decision rule under the DGF(L) policy are similar to that under the DGF policy:
\beq
\label{eq:stopping_policy2}
\bea{l}
\displaystyle \tau= \inf \left\{n \; : \; \Delta_L S(n)\geq -\log c\right\}\;,
\ena
\eeq
where $\Delta_L S(n)\triangleq S_{m^{(L)}(n)}(n)-S_{m^{(L+1)}(n)}(n)$
and
\beq
\label{eq:decision_policy2}
\displaystyle\delta= (m^{(1)}(\tau), m^{(2)}(\tau), ..., m^{(L)}(\tau)) \;.
\eeq

The selection rule under the DGF(L) policy is more involved and depends on the relative order of $K$ and $L$ (or $M-L$). Specifically,
\beq
\label{eq:selection_policy2}
\displaystyle \phi(n)=
\begin{cases} \phi_g(n) \;\;, \;\;\mbox{if \;$\frac{D(g||f)}{L}\geq\frac{D(f||g)}{M-L}\;,$}   \vspace{0.1cm}\\
\phi_f(n) \;\;, \;\;\mbox{if \;$\frac{D(g||f)}{L}<\frac{D(f||g)}{M-L}\;,$}
\end{cases}
\eeq
where
\beq
\label{eq:selection_policy2g}
\displaystyle \phi_g(n)=
\begin{cases} \left(m^{(1)}(n), m^{(2)}(n), ..., m^{(K)}(n)\right) \;\;, \vspace{0.1cm}\\ \hspace{4cm}
            \mbox{if \;$K\geq L$\;,}   \vspace{0.1cm}\\
\left(m^{(L-K+1)}(n), m^{(L-K+2)}(n), ..., m^{(L)}(n)\right) \;\;, \vspace{0.1cm}\\ \hspace{4cm}
            \mbox{if \;$K<L$\;,}
\end{cases}
\eeq
and
\beq
\label{eq:selection_policy2f}
\displaystyle \phi_f(n)=
\begin{cases} \left(m^{(M-K+1)}(n), m^{(M-K+2)}(n), ..., m^{(M)}(n)\right) \;\;,\vspace{0.1cm}\\ \hspace{4cm}
\mbox{if \;$K> M-L$\;,}
            \vspace{0.3cm}\\
              \left(m^{(L+1)}(n), m^{(L+2)}(n), ..., m^{(L+K)}(n)\right) \;\;,\vspace{0.1cm}\\ \hspace{4cm}
              \mbox{if \;$K\leq M-L$\;.}
\end{cases}
\eeq
It is not difficult to see that when $L=1$, the DGF(L) policy degenerates to the DGF policy.

Next, we analyze the performance of the DGF(L) policy. Let
\beq
\bea{l}
\displaystyle I^*(M,K,L)\triangleq \vspace{0.2cm}\\ \hspace{0.5cm}
\begin{cases}
\displaystyle I_g^*(M,K,L)\;\;,\;\;\mbox{if \;$\frac{D(g||f)}{L}\geq\frac{D(f||g)}{M-L}$\;,}   \vspace{0.1cm}\\
\displaystyle I_f^*(M,K,L)\;\;,\;\;\mbox{if \;$\frac{D(g||f)}{L}<\frac{D(f||g)}{M-L}$\;,}
        \end{cases}
\ena
\eeq
where
\beq
\bea{l}
\displaystyle I_g^*(M,K,L)\triangleq \vspace{0.2cm}\\ \hspace{0.5cm}
\begin{cases}
\displaystyle D(g||f)+\frac{(K-L)D(f||g)}{M-L} \;\;, \vspace{0.1cm}\\ \hspace{4cm}
            \mbox{if \;$K\geq L$\;,}   \vspace{0.1cm}\\
\displaystyle \frac{K D(g||f)}{L} \;\;, \vspace{0.1cm}\\ \hspace{4cm}
            \mbox{if \;$K<L$\;,}
        \end{cases}
\ena
\eeq
and
\beq
\bea{l}
\displaystyle I_f^*(M,K,L)\triangleq \vspace{0.2cm}\\ \hspace{0.5cm}
\begin{cases}
\displaystyle D(f||g)+\frac{(K-M+L) D(g||f)}{L} \;\;,\vspace{0.1cm}\\ \hspace{4cm} \mbox{if \;$K> M-L$\;,}
            \vspace{0.3cm}\\
\displaystyle\frac{K D(f||g)}{M-L} \;\;,\vspace{0.1cm}\\ \hspace{4cm} \mbox{if \;$K\leq M-L$\;.}
        \end{cases}
\ena
\eeq

The following theorem shows the asymptotically optimal performance of the DGF(L) policy:\vspace{0.2cm}
\begin{theorem}
\label{th:optimality_policy2}
Let $R^*$ and $R(\Gamma)$ be the Bayes risks under the DGF(L) policy and any other policy $\Gamma$, respectively. Then,
\beq
\label{eq:asymptotic_performance_policy2}
R^* \;\sim\; \frac{-c\log c}{I^*(M,K,L)}\;\sim\;\inf_{\Gamma}\;{R(\Gamma)} \;\;\; \mbox{as} \;\;\; c\rightarrow 0 \;.
\vspace{0.2cm}
\eeq
\end{theorem}
%
\begin{proof}
See Appendix~\ref{app:proof_policy2}.
\end{proof}
\vspace{0.2cm}

Note that in the DGF(L) policy, all $L$ targets are declared simultaneously at the termination time of the detection. A modification to DGF(L) leads to a policy where abnormal processes are declared sequentially during the detection. Consider, for example, $K=1$ and $\frac{D(g||f)}{L}\geq\frac{D(f||g)}{M-L}$. It can be shown (with minor modifications to Theorem~\ref{th:optimality_policy2}) that an asymptotically optimal policy is to test the cell with the largest sum LLRs and declare the first target once the largest sum LLRs exceeds the threshold $-\log c$. The same procedure is then applied to the remaining $M-1$ cells. This repeats until $L$ abnormal processes have been declared, at which point, the detection terminates. The asymptotic expected termination time is given by $-L\log c/D(g||f)$ with $P_e=O(c)$. Even though the total detection time remains the same order as under the DGF(L) policy, this modified version may be more appealing from a practical point of view. In particular, actions can be taken to fix each abnormal process the moment it is identified; the total impact to the system by these $L$ abnormal processes can thus be reduced.
If $\frac{D(g||f)}{L}<\frac{D(f||g)}{M-L}$, it can be shown that an asymptotically optimal policy is to test the cell with the smallest sum LLRs and declare the first \emph{normal} process once the smallest sum LLRs drops below $\log c$. The same procedure is then applied to the remaining $M-1$ processes and is repeated until all $M-L$ objects are declared as normal (thus, the $L$ remaining ones are declared as abnormal). The asymptotic expected termination time is given by $-(M-L)\log c/D(f||g)$ with $P_e=O(c)$.
Even though in this case, the modified version also declares all $L$ targets simultaneously at the termination time of the detection, the difference is that this modified version incurs much few switchings among processes than the DGF(L) policy. This may be more advantageous in some practical scenarios when switching among tested processes results in additional cost or delay. To see that the modified version incurs few switchings, we note that the modified version tests the process that the decision maker is most sure to be normal based on past observations while DGF(L) tests the process that the decision maker is least sure to be normal except the $L$ processes currently considered as the targets (see the second line in (\ref{eq:selection_policy2f}) which shows that DGF(L) chooses the cell with the $(L+1)^{th}$ largest sum LLRs; the $L$ processes with larger sum LLRs are the current maximum likelihood of the target locations). It should be noted that those modified DGF(L) schemes are expected to achieve the same performance as DGF(L) in both the finite and asymptotic regimes when $K=1$ following a similar argument as in Section~\ref{ssec:comparison}.

\subsection{Unknown Number of Abnormal Processes}
\label{ssec:upper_bound_L}

In this section we consider the interesting case where the number $\ell$ of abnormal processes (or targets) is unknown. It is only known that $\ell$ is bounded by $1\leq\ell\leq L$. We consider the case where $K=1$.
We also assume that the number of cells satisfies:
\beq\label{eq:M}
M\geq\frac{L\left(D(g||f)+D(f||g)\right)}{D(g||f)}\;.
\eeq
Note that (\ref{eq:M}) implies $\frac{D(g||f)}{L}\geq\frac{D(f||g)}{M-L}$.

Throughout this section, we allow the decision maker to declare the target locations sequentially during the test (similar to the modified DGF(L) policy as discussed at the end of Section~\ref{ssec:L}). We refer to the detection time~$\tau_d$ as the time when the last target has been declared and to the termination time~$\tau$ as the time when the decision maker terminates the test. Note that $\tau=\tau_d$ when the number $\ell$ of targets is known (as discussed in previous sections). When $\ell$ is unknown, however, $\tau_d\leq\tau$ since the decision maker does not know whether it has already identified all targets at time~$\tau_d$. In general, the termination time~$\tau$ increases linearly with $M$ under any policy with $P_e=O(c)$ whenever $\ell<L$. This is due to the fact that even if the $\ell$ targets have been detected with sufficient reliability, the decision maker must verify whether there are other targets in the remaining $M-\ell$ cells before terminating the test. On the other hand, following the modified DGF(L) policy, one would expect to achieve a detection time~$\tau_d$ less than $-L\log c/D(g||f)$ for all $\ell\leq L$, which is independent of the number $M$ of total processes.

In scenarios with a large number of processes and $L<<M$, a policy that focuses on minimizing the termination time~$\tau$, which grows linearly with $M$, may not be practically appealing. It is desirable to have a policy that allows each abnormal process to be identified and fixed as quickly as possible \emph{during} the test. In other words, it is desirable to have a policy that minimizes the detection time~$\tau_d$ rather than the termination time~$\tau$. In this case, even though the test continues after $\tau_d$ to ensure there are no other targets, all abnormal processes have been fixed by the detection time~$\tau_d$ and cease to incur cost to the system. We thus modify the objective function to the following Bayes risk:
\beq
\label{eq:Bayes_risk_D}
\displaystyle R(\Gamma)\triangleq P_e(\Gamma)+c\mathbf{E}(\tau_d|\Gamma) \;,
\eeq
and we are interested in finding a strategy $\Gamma$ that minimizes the Bayes risk (\ref{eq:Bayes_risk_D}) This design objective is similar to that considered in~\cite{Cohen_2013_Optimal_GlobalSIP, Cohen_2014_Optimal, Cohen_2014_Asymptotically}.

Before presenting the desired solution for this case, we demonstrate with a specific example that even though the Chernoff test is asymptotically optimal in terms of minimizing the termination time~$\tau$, it is highly suboptimal in terms of minimizing the detection time~$\tau_d$.
Assume that $L=2$ is the upper bound on the number of targets, which can locate in any of $M=3$ cells. As a result, the detection problem includes $6$ hypotheses, $H_1=\left\{1\right\}, H_2=\left\{2\right\}, H_3=\left\{3\right\}, H_4=\left\{1,2\right\}, H_5=\left\{1,3\right\}, H_6=\left\{2,3\right\}$. The observation model under every hypothesis and cell selection is given in Table~\ref{tab:observation}.
\begin{table}
\caption{Observation Model}
	\centering
		\small\begin{tabular}{|c|c|c|c|}
			\hline  & cell $1$ & cell $2$ & cell $3$ \\  \hline
             $H_1=\left\{1\right\}$ & g & f & f   \\ \hline
             $H_2=\left\{2\right\}$ & f & g & f   \\ \hline
             $H_3=\left\{3\right\}$ & f & f & g \\ \hline
             $H_4=\left\{1,2\right\}$ & g & g & f \\ \hline
             $H_5=\left\{1,3\right\}$ & g & f & g \\ \hline
             $H_6=\left\{2,3\right\}$ & f & g & g \\
      \hline  			
    \end{tabular}
    \label{tab:observation}
\end{table}
Assume that hypothesis $H_1$ is true and that $\hat{H}(n)=H_1$, where $\hat{H}(n)$ is the ML estimate of the true hypothesis at time~$n$. Consider a deterministic policy that selects the cells according to the order of their sum LLRs and declares an object as target if $S_m(n)>-\log c$ or normal if $S_m(n)<\log c$. This policy achieves $\tau_d\sim-\log c/D(g||f)$ (since cell $1$ is first identified as a target with high probability) and $\tau\sim-\log c/D(g||f)-2\log c/D(f||g)$ (since the number of targets is unknown and $L=2$, thus the decision maker must continue testing the normal processes before terminating the test). On the other hand, the Chernoff test (which aims to minimize the termination time) will not select cell $1$ at time~$n$, since $H_4$ or $H_5$ minimizes (\ref{eq:selection_Chernoff}), i.e., $D(p_{H_1}^1||p_{H_4}^1)=D(p_{H_1}^1||p_{H_5}^1)=D(g||g)=0$. It can be verified that selecting randomly cells $2$ or $3$ with equal probability $1/(M-1)=1/2$ maximizes (\ref{eq:selection_Chernoff}) and achieves a rate function $D(f||g)/(M-1)=D(f||g)/2$, which results in $\tau=\tau_d\sim-2\log c/D(f||g)$, which is greater than the detection time under the above deterministic policy. Intuitively speaking, once cells $2, 3$ are identified as normal, cell $1$ is identified as abnormal (because at least $1$ target is present). Therefore, the Chernoff test observes cells $2, 3$ to minimize the termination time~$\tau$ (by not testing cell $1$), while increasing the detection time~$\tau_d$.

Next, we present a deterministic policy to minimize the Bayes risk (\ref{eq:Bayes_risk_D}). Let $\mathcal{T}(n)$ be the set of cells satisfying $S_m(n)\geq-\log c$ at time~$n$. Define
\beq
\tilde{m}^{(1)}(n)=\arg\;\max_{m\nin\mathcal{T}(n)}S_m(n) \;.
\eeq
The selection rule is given by:
\beq
\label{eq:selection_policy3}
\displaystyle \phi(n)=\tilde{m}^{(1)}(n) \;.
\eeq
The stopping rule and decision rule are given by:
\beq
\label{eq:stopping_policy3}
\bea{l}
\displaystyle \tau= \inf \left\{n \; : \; |S_m(n)|\geq -\log c \;\; \forall m\right\} \;,
\ena
\eeq
and
\beq
\label{eq:decision_policy3}
\displaystyle\delta= \mathcal{T}(\tau_d) \;.
\eeq
Note that $\delta$ denotes the target locations and the complete set is declared at time~$\tau_d$. Since the number of targets is unknown, the decision maker continues taking observations to verify that there is no other target. The test is terminated at time~$\tau$.

The following theorem shows the asymptotically optimal performance of the proposed policy:\vspace{0.2cm}

\begin{theorem}
\label{th:optimality_policy3}
Let $\ell\leq L$ be the number of targets, $K=1$ and assume that (\ref{eq:M}) holds. Let $R^*$ and $R(\Gamma)$ be the Bayes risks (\ref{eq:Bayes_risk_D}) under the proposed policy and any other policy $\Gamma$, respectively. Then,
\beq
\label{eq:asymptotic_performance_policy3}
R^* \;\sim\; \frac{-\ell c\log c}{D(g||f)}\;\sim\;\inf_{\Gamma}\;{R(\Gamma)} \;\;\; \mbox{as} \;\;\; c\rightarrow 0 \;.
\vspace{0.2cm}
\eeq
\end{theorem}
\begin{proof}
See Appendix~\ref{app:proof_policy3}.
\end{proof}
\vspace{0.2cm}

\section{Numerical Examples}
\label{sec:simulation}

In this section we present numerical examples to illustrate the performance of the proposed deterministic policy as compared to the Chernoff test. We simulated a single anomalous object (i.e., target) located in one of $M$ cells with the following parameters: The \emph{a priori} probability that the target is present in cell $m$ was set to $\pi_m=1/M$ for all $1\leq m\leq M$.
When cell $m$ is observed at time~$n$, an observation $y_m(n)$ is independently drawn from a distribution $f\sim\exp(\lambda_f)$ or $g\sim\exp(\lambda_g)$, depending on whether the target is absent or present, respectively. It can be verified that:
\begin{center}
$\bea{l}
\displaystyle D(g||f)=\log(\lambda_g)-\log(\lambda_f)+\frac{\lambda_f}{\lambda_g}-1\;, \vspace{0.1cm} \\
\displaystyle D(f||g)=\log(\lambda_f)-\log(\lambda_g)+\frac{\lambda_g}{\lambda_f}-1 \;.
\ena$
\end{center}
Let $R_{DGF}, R_{Ch}$ be the Bayes risks under the DGF policy and the Chernoff test, respectively. Let $R_{LB}=\frac{-c\log c}{I^*(M,K)}$ be the asymptotic lower bound on the Bayes risk as $c\rightarrow 0$. We define:
\begin{center}
$\bea{l}
\displaystyle L_{DGF}\triangleq \frac{R_{DGF}-R_{LB}}{R_{LB}} \;,\vspace{0.2cm}\\
\displaystyle L_{Ch}\triangleq \frac{R_{Ch}-R_{LB}}{R_{LB}} \;.
\ena$
\end{center}
as the relative loss in terms of Bayes risk under the DGF policy and the Chernoff test, respectively, as compared to the asymptotic lower bound. Following Theorems~\ref{th:optimality_Chernoff},~\ref{th:optimality_policy1}, we expect both $L_{DGF}$ and $L_{Ch}$ to approach $0$ as $c\rightarrow 0$. $L_{DGF}$ and $L_{Ch}$ serve as performance measures of the tests in the finite regime.

First, we consider the case where $M=5$ and $K=1$. Note that when $K=1$ and $D(g||f)\geq D(f||g)/(M-1)$, the Chernoff test coincides with the DGF policy: they both select cell $m^{(1)}(n)$. When $D(g||f)<D(f||g)/(M-1)$, however, the proposed policy selects cell $m^{(2)}(n)$, while the Chernoff test selects cell $j\neq m^{(1)}(n)$ randomly at each given time~$n$. We set $\lambda_f=0.5, \lambda_g=10$ and obtain $D(g||f)\approx 2.05, D(f||g)/(M-1)\approx 4$. As a result, the Chernoff test and the DGF policy have different cell selection rules.
The performance of the Algorithms is presented in Fig.~\ref{fig:fig2a},~\ref{fig:fig2b}, were $10^7$ trials were performed. In Fig.~\ref{fig:fig2a}, the asymptotic lower bound on the expected sample size and the average sample sizes achieved by the algorithms are presented as a function of $c$ (log-scale). In Table~\ref{tab:std_confidnce} we present the sample standard deviations $\sigma$ and the standard deviation multipliers $r$ for a 95\% confidence intervals $[\bar{\tau}-r\sigma,\bar{\tau}+r\sigma]$, where $\bar{\tau}$ is the average detection delay. In Fig.~\ref{fig:fig2b}, $L_{DGF}$ and $L_{Ch}$ are presented as a function of $c$. Although both schemes approach the asymptotic lower bound as $c\rightarrow 0$, it can be seen that the DGF policy significantly outperforms the Chernoff test in the finite regime for all values of $c$.
\begin{figure}[t] 
\begin{center}
    \subfigure[Average sample sizes and the asymptotic lower bound as a function of the cost per observation.]{\scalebox{0.45}
    {
      \label{fig:fig2a}{\epsfig{file=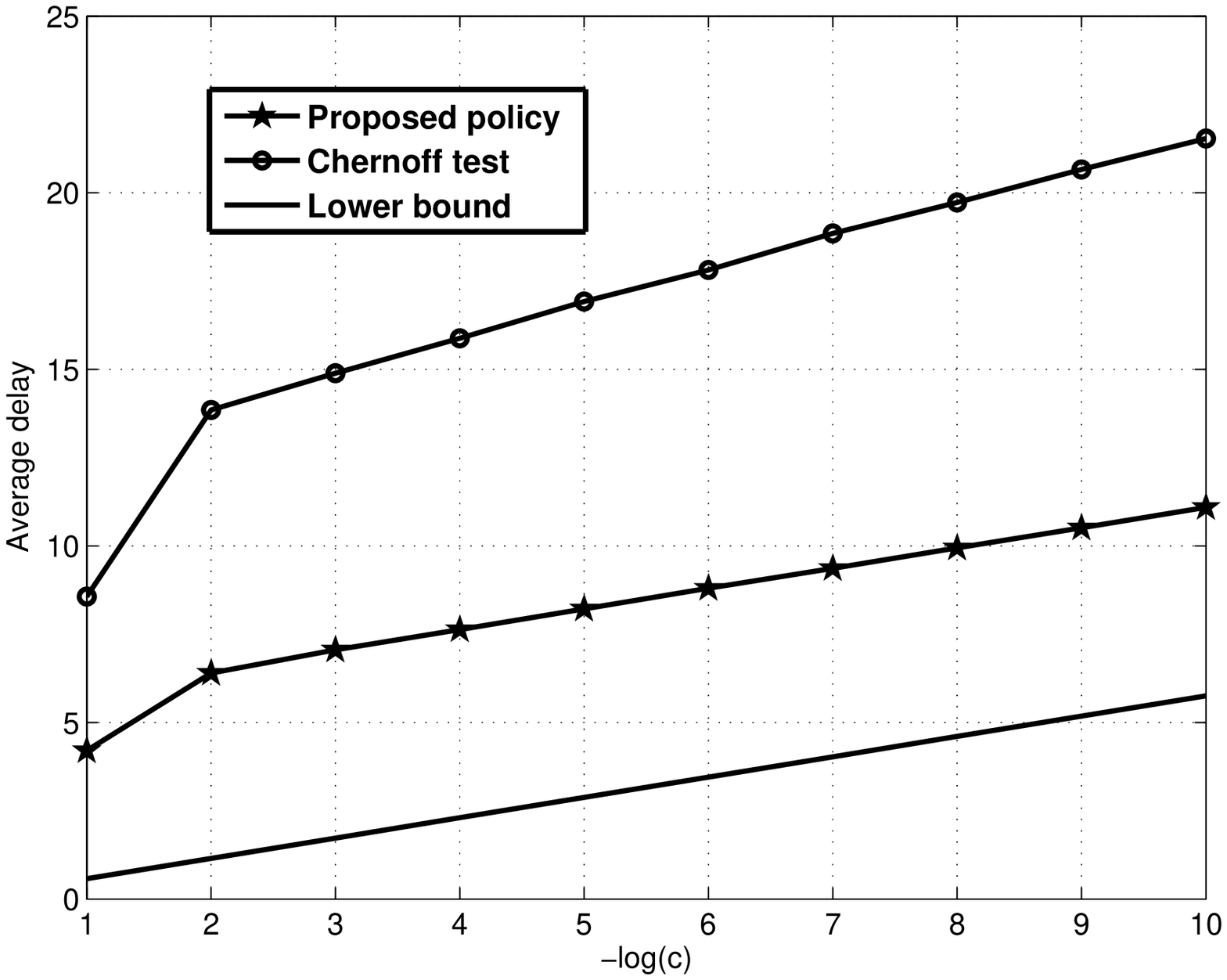}}
    }}
    \subfigure[The loss in terms of Bayes risk under the DGF policy and the Chernoff test as compared to the asymptotic lower bound. $L_{DGF}, L_{Ch}$ approach $0$ as $c\rightarrow 0$]{\scalebox{0.45}
    {
      \label{fig:fig2b}{\epsfig{file=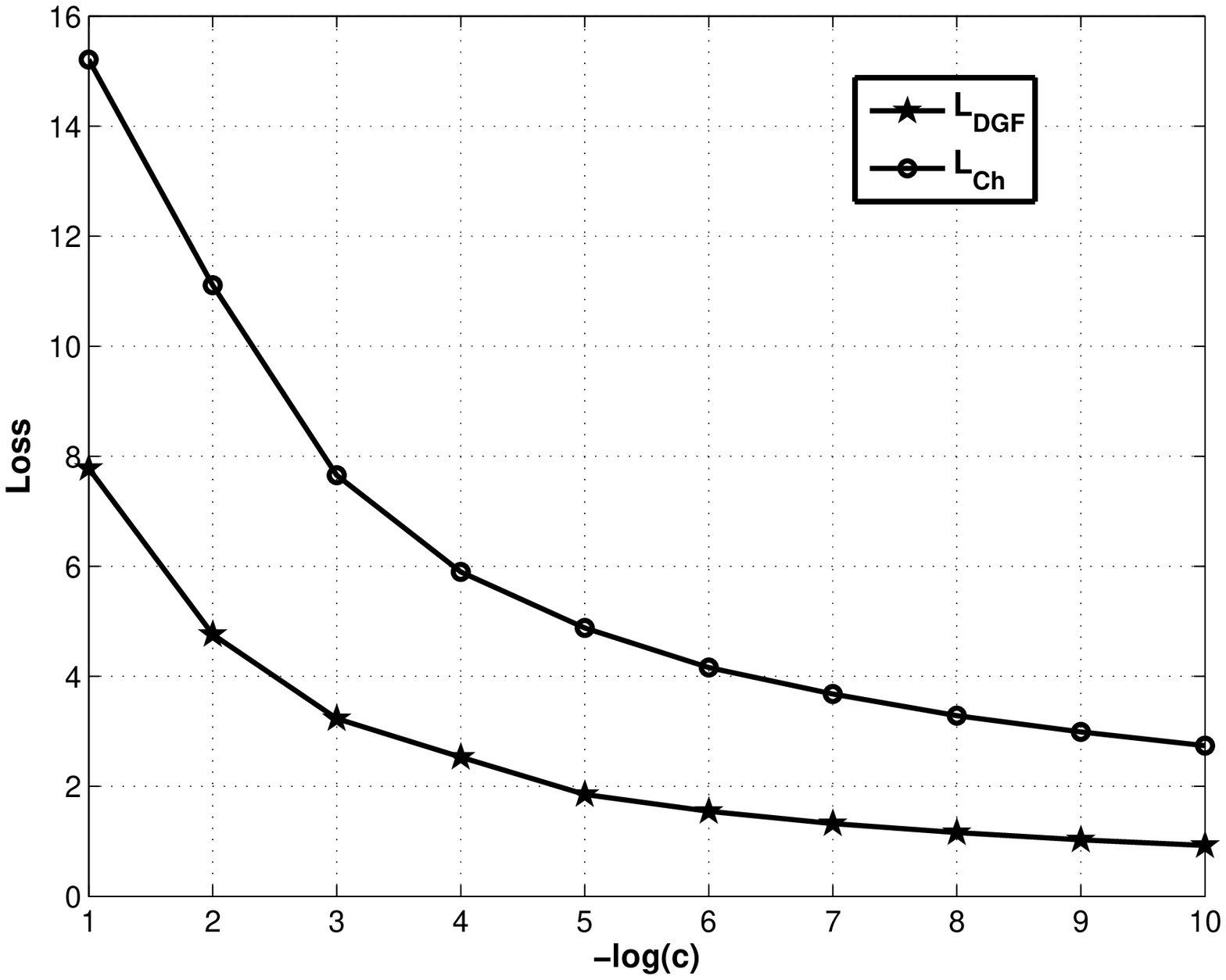}}
    }}
   \caption{Performance comparison for $M=5$, $K=1$, $\lambda_f=0.5, \lambda_g=10$}
  \label{fig:fig2}
\end{center}
  \end{figure}

\begin{table}
\caption{Values for 95\% Confidence Level}
	\centering
		\small\begin{tabular}{|c|c|c|}
	\hline $-\log c$ & DGF        & Chernoff Test  \\  \hline
         $1$ & $\sigma=1.84\;,\;r=1.60$ & $\sigma=6.75\;,\;r=1.90$  \\ \hline
         $3$ & $\sigma=2.00\;,\;r=1.80$ & $\sigma=7.17\;,\;r=1.95$
         \\  \hline
         $5$ & $\sigma=2.24\;,\;r=1.90$ & $\sigma=7.75\;,\;r=1.95$
         \\  \hline
    \end{tabular}
    \label{tab:std_confidnce}
\end{table}

Next, we consider the case where $M=5$ and $K=2$ (i.e., two cells are observed at a time). In this case, the DGF policy selects cells $m^{(1)}(n)$ and $m^{(2)}(n)$ at each given time~$n$ only if $D(g||f)\geq D(f||g)/(M-1)$. Otherwise, it selects cells $m^{(2)}(n)$ and $m^{(3)}(n)$. The Chernoff test selects cells $m^{(1)}(n)$ and $j\neq m^{(1)}(n)$ (randomly) at each given time~$n$ only if $D(g||f)\geq D(f||g)/(M-1)$. Otherwise, it selects cells $i, j\neq m^{(1)}(n)$ randomly. First, we set $\lambda_f=2, \lambda_g=10$ and obtain $D(g||f)\approx 0.8, D(f||g)/(M-1)\approx 0.6$. The performance of the algorithms is presented in Fig.~\ref{fig:fig3a},~\ref{fig:fig3b}. Next, we set $\lambda_f=0.5$, $\lambda_g=10$ and obtain $D(g||f)\approx 2.05, D(f||g)/(M-1)\approx 4$. The performance of the algorithms in this case is presented in Fig.~\ref{fig:fig4a},~\ref{fig:fig4b}. In Fig.~\ref{fig:fig3a},~\ref{fig:fig4a}, the asymptotic lower bound on the expected sample size and the average sample sizes achieved by the algorithms are presented as a function of the cost per observation $c$. In Fig.~\ref{fig:fig3b},~\ref{fig:fig4b}, $L_{DGF}$ and $L_{Ch}$ are presented as a function of $c$.
\begin{figure}[t] 
\begin{center}
    \subfigure[Average sample sizes and the asymptotic lower bound as a function of the cost per observation.]{\scalebox{0.45}
    {
      \label{fig:fig3a}{\epsfig{file=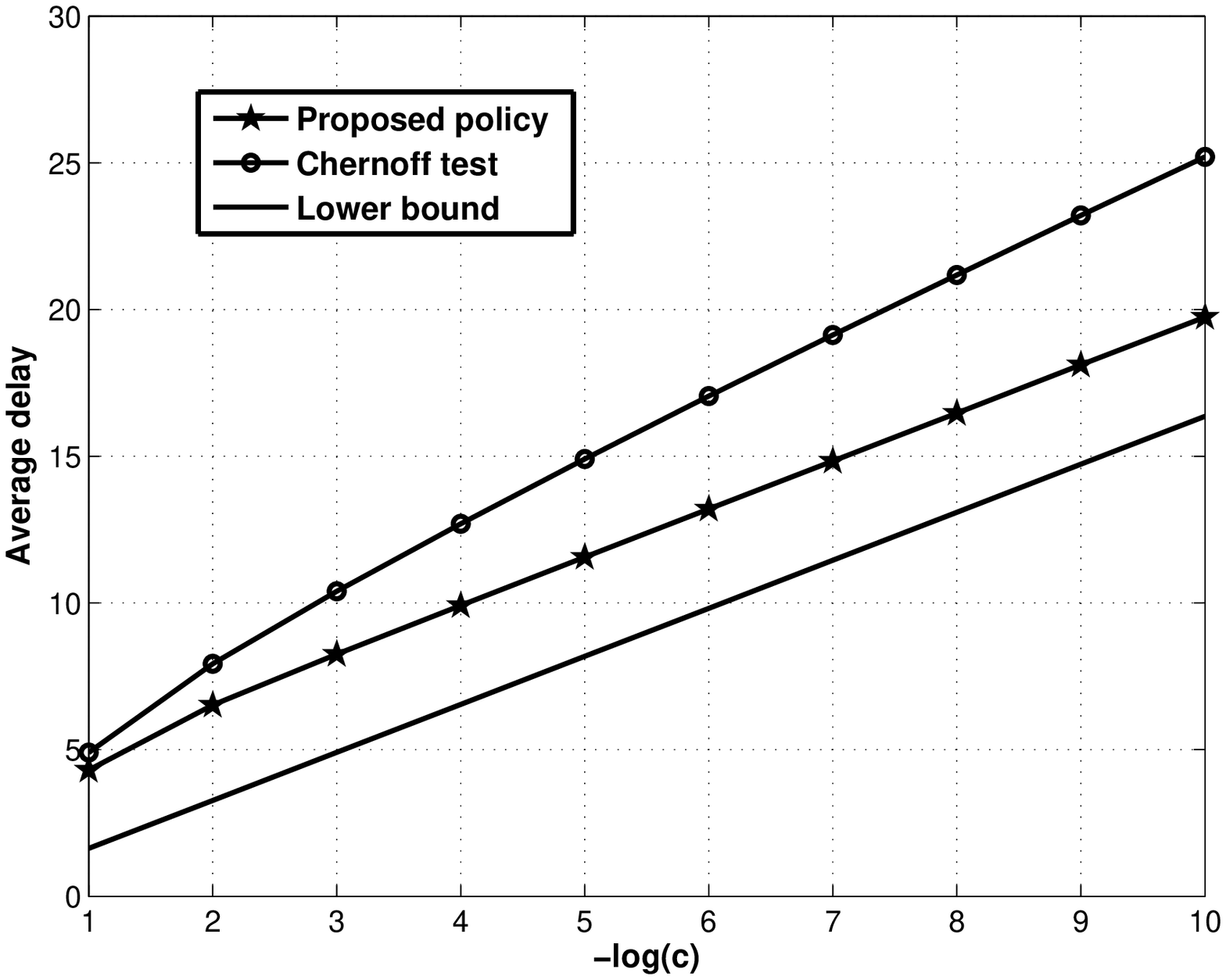}}
    }}
    \subfigure[The loss in terms of Bayes risk under the DGF policy and the Chernoff test as compared to the asymptotic lower bound. $L_{DGF}, L_{Ch}$ approach $0$ as $c\rightarrow 0$]{\scalebox{0.45}
    {
      \label{fig:fig3b}{\epsfig{file=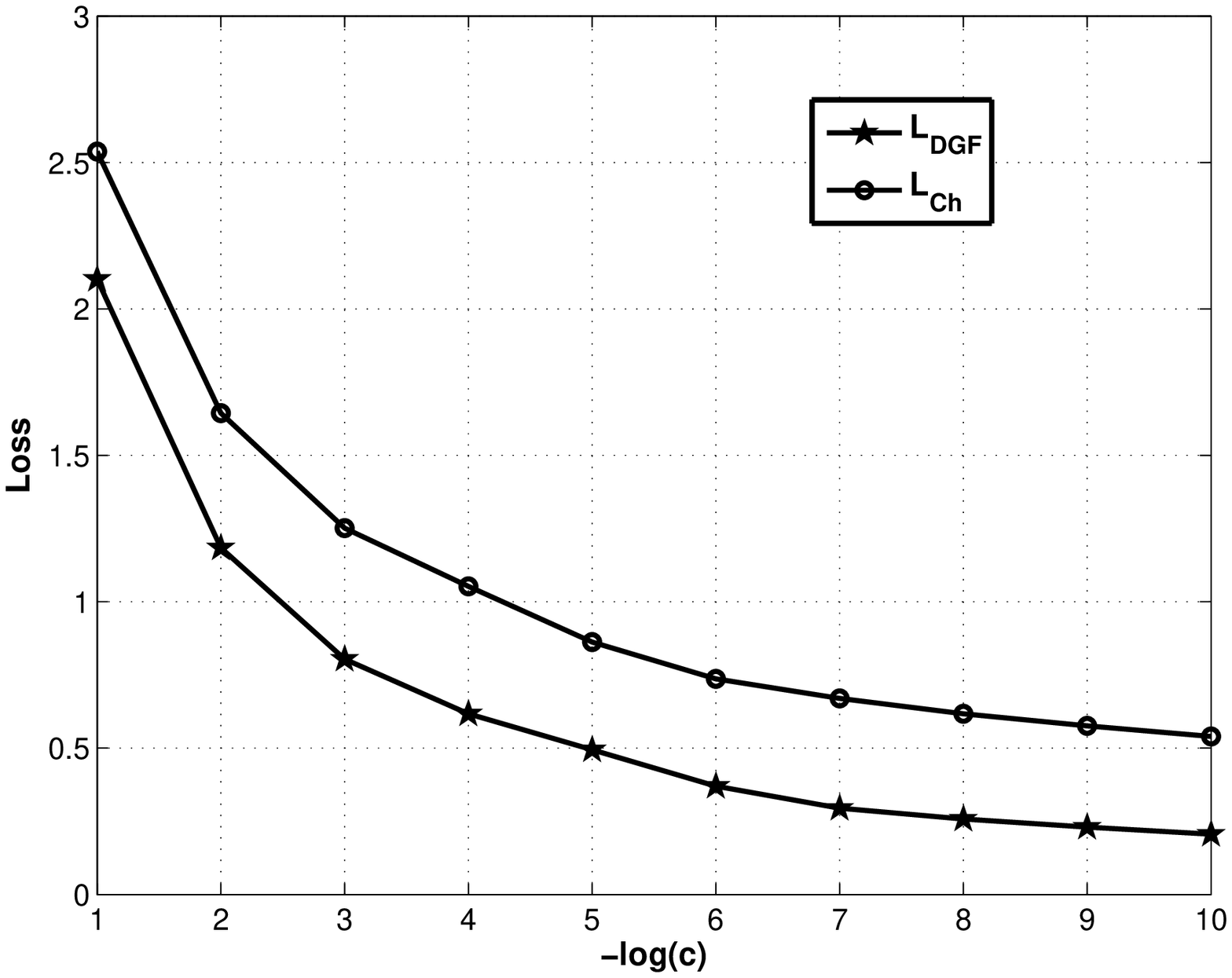}}
    }}
   \caption{Performance comparison for $M=5$, $K=2$, $\lambda_f=2, \lambda_g=10$}
  \label{fig:fig3}
\end{center}
  \end{figure}
It can be seen that the DGF policy significantly outperforms the Chernoff test in the finite regime for all values of $c$ under all cases. These results demonstrate the advantage of using the deterministic selection rule applied by the DGF policy instead of the randomized Chernoff test for the anomaly detection problem.
\begin{figure}[t] 
\begin{center}
    \subfigure[Average sample sizes and the asymptotic lower bound as a function of the cost per observation.]{\scalebox{0.45}
    {
      \label{fig:fig4a}{\epsfig{file=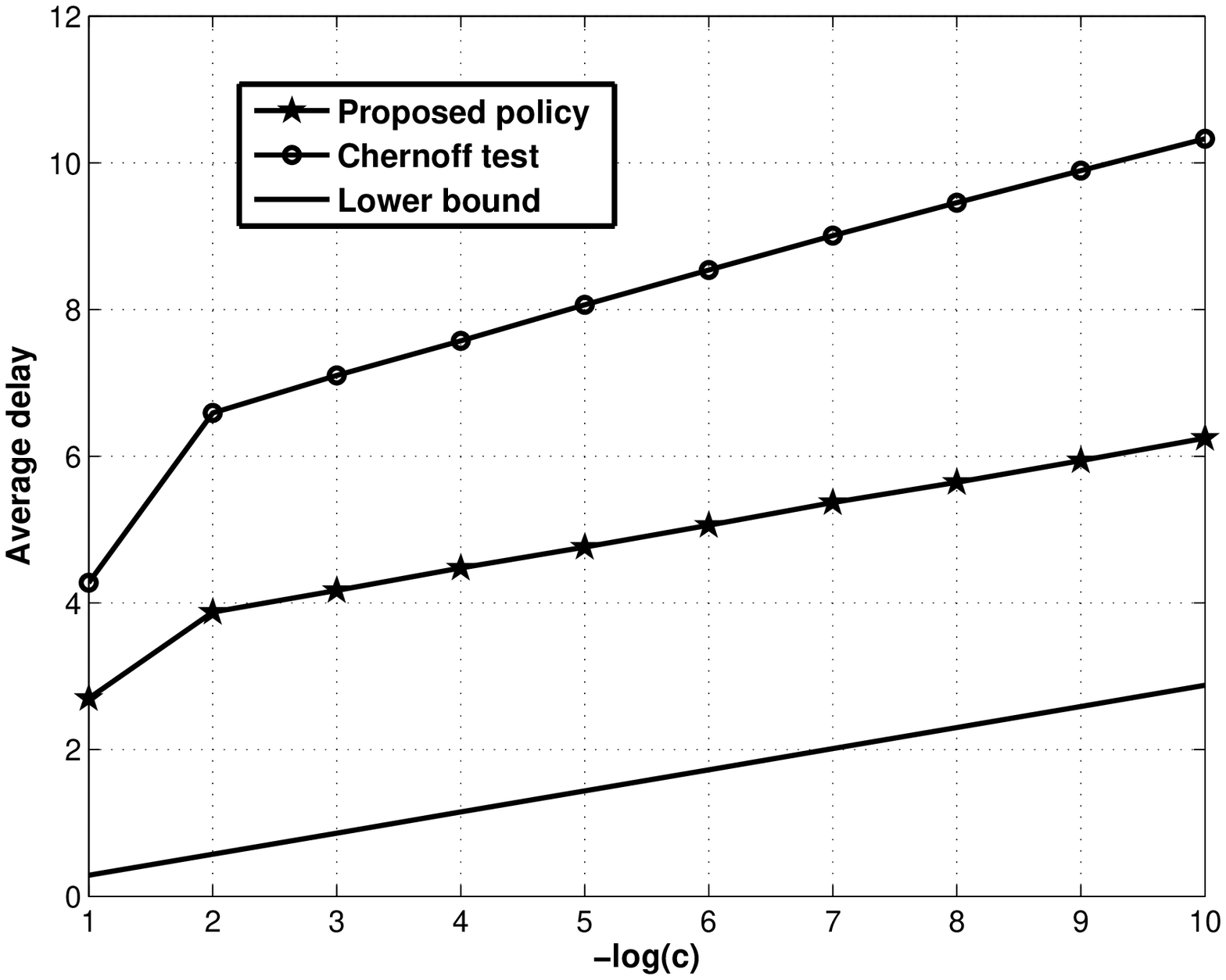}}
    }}
    \subfigure[The loss in terms of Bayes risk under the DGF policy and the Chernoff test as compared to the asymptotic lower bound. $L_{DGF}, L_{Ch}$ approach $0$ as $c\rightarrow 0$]{\scalebox{0.45}
    {
      \label{fig:fig4b}{\epsfig{file=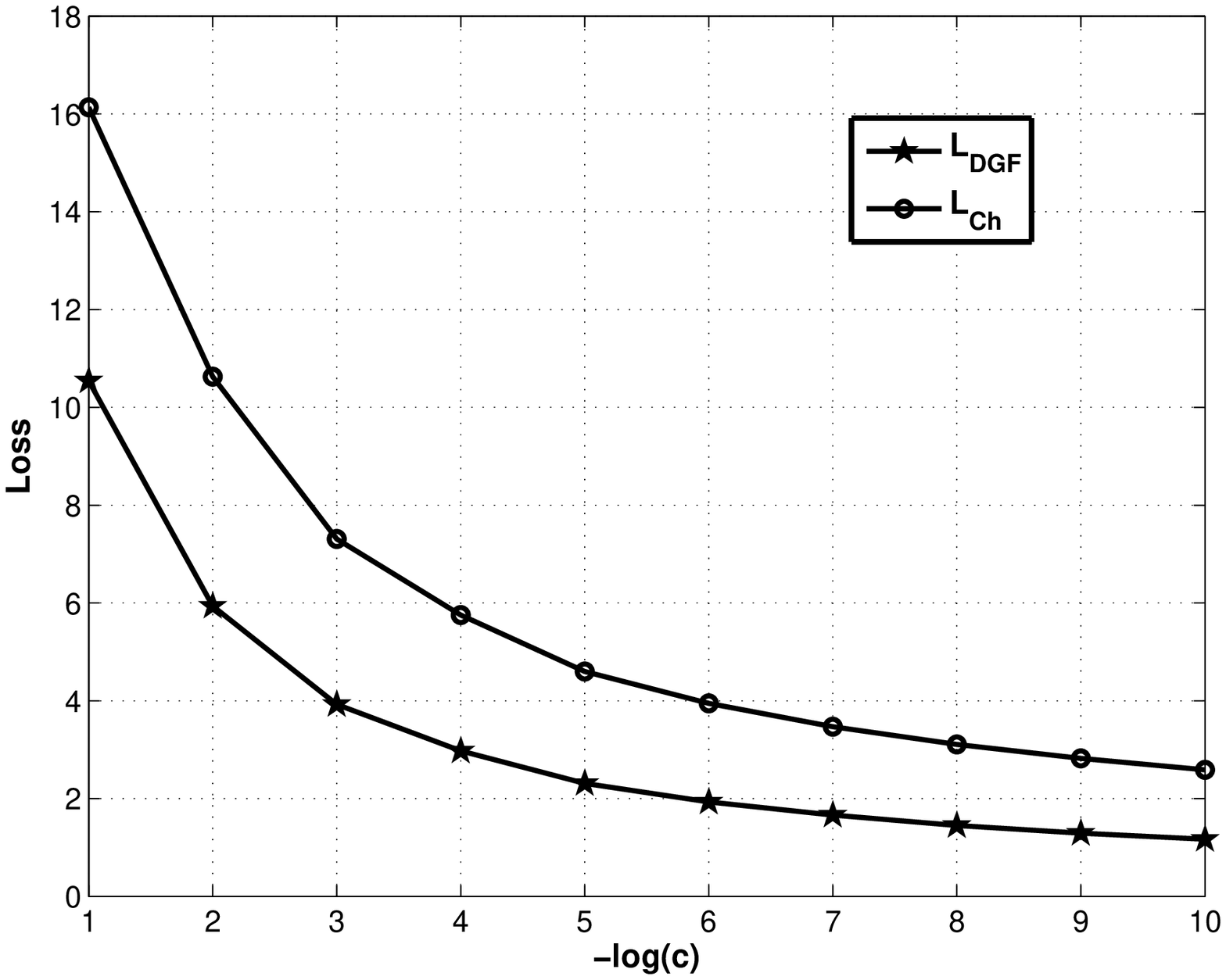}}
    }}
   \caption{Performance comparison for $M=5$, $K=2$, $\lambda_f=0.5, \lambda_g=10$}
  \label{fig:fig4}
\end{center}
  \end{figure}

\section{Conclusion}
\label{sec:conclusion}

The problem of quickest detection of an anomalous process (i.e., target) among $M$ processes (i.e., cells) was investigated. Due to resource constraints, only a subset of the cells can be observed at a time, The objective is a search strategy that minimizes the expected search time subject to an error probability constraint. The observations from searching a cell are realizations drawn from two different distributions $f$ or $g$, depending on whether the target is absent or present, respectively. A simple deterministic policy was established to solve the Bayesian formulation of the search problem, where a cost of $c$ per observation and a loss of $1$ for wrong decisions are assigned. It is shown that the proposed index policy is asymptotically optimal in terms of minimizing the Bayes risk as $c$ approaches zero.

The problem was further extended to handle the case where multiple anomalous processes are present. In particular, the interesting case where only an upper bound on the number of anomalous processes is known was considered. We showed that existing methods may not be practically appealing under the latter setting. Hence, we proposed a modified optimization problem for this case. Asymptotically optimal deterministic policies were developed for these cases as well.

\section{Appendix}
\label{app}

\subsection{Proof of Theorem~\ref{th:optimality_policy1}}
\label{app:proof_policy1}

In this appendix we prove the asymptotic optimality of the DGF policy as $c\rightarrow 0$. In App.~\ref{ssec:lower_bound}, we show that $\frac{-c\log c}{I^*(M,K)}$ is an asymptotic lower bound on the Bayes risk that can be achieved by any policy $\Gamma$. Then, we show in App.~\ref{ssec:asymptotic_policy1} that the Bayes risk $R^*$ under the DGF policy, approaches the asymptotic lower bound as $c\rightarrow 0$. Specifically, the asymptotic optimality property of DGF is based on Lemma~\ref{lemma:expected_time_policy1}, showing that the asymptotic expected search time approaches $\frac{-\log c}{I^*(M,K)}$, while the error probability is $O(c)$ following Lemma~\ref{lemma:error_policy1}.

Throughout the appendix we use the following notations: Let
\beq
N_j(n)\triangleq\sum_{t=1}^{n}{\mathbf{1}_j(t)}
\eeq
be the number of times that cell $j$ has been observed up to time~$n$.\\
We define
\beq
\label{eq:Delta_S_m_j}
\Delta S_{m,j}(n)\triangleq S_m(n)-S_j(n) \;,
\eeq
as the difference between the observed sum of LLRs of cells $m$ and $j$. Let
\beq
\label{eq:Delta_S_m}
\Delta S_m(n)\triangleq\min_{j\neq m} \Delta S_{m,j}(n) \;.
\eeq
Thus,
\beq
\label{eq:Delta_S}
\Delta S(n)= S_{m^{(1)}(n)}(n)-S_{m^{(2)}(n)}(n)=\max_{m} \Delta S_m(n) \;.
\eeq

Without loss of generality we prove the theorem when hypothesis $m$ is true. For convenience, we define
\beq
\displaystyle \tilde{\ell}_k(i)=
\begin{cases} \ell_k(i)-D(g||f) \;,\;
                                                \mbox{if $k=m$,}   \vspace{0.3cm}\\
              \ell_k(i)+D(f||g) \;,\;
              \mbox{if $k\neq m$.}
\end{cases}
\eeq
Note that $\tilde{\ell}_k(i)$ is a zero-mean r.v under hypothesis $H_m$. \vspace{0.2cm}

\subsubsection{The Asymptotic Lower bound on the Bayes risk}
\label{ssec:lower_bound}
\hspace{0.0cm}\vspace{0.2cm}\\
The asymptotic lower bound on the Bayes risk is shown in Theorem~\ref{th:lower_bound} below and is mainly based on Lemmas~\ref{lemma:Delta_S_O},~\ref{lemma:Delta_S}, provided below. Throughout this section, $\tau$ denotes a generic stopping time that can be determined by any policy $\Gamma$. In Section~\ref{ssec:asymptotic_policy1}, however, we will refer to $\tau$ as the specific stopping time under the DGF policy.
Lemma~\ref{lemma:Delta_S_O} shows that under $\mathbf{P}_m$, $\Delta S_m(\tau)$, defined in (\ref{eq:Delta_S_m}), must be large enough to obtain a sufficiently small error $\alpha_m$. Lemma~\ref{lemma:Delta_S} implies that $\tau$ must be large enough to obtain a sufficiently large $\Delta S_m(\tau)$.
 \vspace{0.2cm}
\begin{lemma}
\label{lemma:Delta_S_O}
Assume that $\alpha_j(\Gamma)=O(-c\log c)$ for all $j=1, ..., M$. Let $0<\epsilon<1$. Then:
\beq
\mathbf{P}_m\left(\Delta S_m(\tau)< -\left(1-\epsilon\right)\log c\;|\;\Gamma\right)=O(-c^{\epsilon}\log c) \;,
\eeq
for all $m=1, ..., M$.
\vspace{0.2cm} \\
\end{lemma}
\begin{proof}
Note that:
\beq
\bea{l}
\mathbf{P}_m\left(\Delta S_m(\tau)< -\left(1-\epsilon\right)\log c|\Gamma\right) \vspace{0.2cm} \\
=\mathbf{P}_m\left(\Delta S_m(\tau)< -\left(1-\epsilon\right)\log c \;,\; \delta=m|\Gamma\right) \vspace{0.2cm} \\
+\mathbf{P}_m\left(\Delta S_m(\tau)< -\left(1-\epsilon\right)\log c \;,\; \delta\neq m|\Gamma\right) \vspace{0.2cm} \\
\leq\mathbf{P}_m\left(\Delta S_m(\tau)< -\left(1-\epsilon\right)\log c \;,\; \delta=m|\Gamma\right)+\alpha_m(\Gamma) \vspace{0.2cm} \\
\ena
\eeq
Note that $\alpha_m(\Gamma)=O(-c\log c)$ as conditioned by the Lemma. Next, we upper bound the term \\$\mathbf{P}_m\left(\Delta S_m(\tau)< -\left(1-\epsilon\right)\log c \;,\; \delta=m|\Gamma\right)$ By changing the measure, as in~\cite[Lemma $4$]{Chernoff_1959_Sequential}.

Let $R_{\tau}$ be the subset of the sample space, in which $\Delta S_{m,j}(\tau)<-(1-\epsilon)\log c$ for some $j\neq m$ and $H_m$ is accepted at time $\tau$.
Let $y_k(i)$ be the observation collected from cell $k$ at time $i$ (note that only $K$ observations are obtained at a time. An observation is meaningful only when the process is probed. Otherwise, we can set an arbitrary value). Let $y(\tau)=\left\{y_1(i), ..., y_M(i)\right\}_{i=1}^{\tau}$ be the set of all the observations up to time $\tau$.
Let $\mathcal{N}_k(y(\tau))$ be the set of time indices for the observations $y(\tau)$, containing the time indices in which cell $k$ was probed.
Thus, for all $j\neq m$ there exists $G>0$ such that:
\beq
\bea{l}
\hspace{-0.5cm}-Gc\log c \geq\mathbf{P}_j\left(\delta\neq j|\Gamma\right)\geq\mathbf{P}_j\left(\delta=m|\Gamma\right)  \vspace{0.2cm} \\
\hspace{-0.5cm}\geq \mathbf{P}_j\left(\Delta S_{m,j}(\tau)\leq -(1-\epsilon)\log c\;,\;\delta=m|\Gamma\right) \vspace{0.2cm} \\
\hspace{-0.5cm}=\displaystyle\sum_{\tau=1}^{\infty}\int_{R_{\tau}} 
\displaystyle\left[\prod_{i\in \mathcal{N}_m(y(\tau))}{\hspace{-0.2cm}f(y_m(i))}\prod_{i\in \mathcal{N}_{j}(y(\tau))}{\hspace{-0.2cm}g(y_j(i))}\times
\vspace{0.2cm}\right.\\\left. \hspace{2cm} \displaystyle
\prod_{k\neq m, j}\;\prod_{i\in \mathcal{N}_k(y(\tau))}{\hspace{-0.3cm}f(y_k(i))}\right]d\mu(y(\tau)) \vspace{0.2cm} \\
\hspace{-0.6cm}=\displaystyle\sum_{\tau=1}^{\infty}\int_{R_{\tau}} 
\displaystyle\left[\prod_{i\in \mathcal{N}_m(y(\tau))}{\frac{f(y_m(i))}{g(y_m(i))}}\prod_{i\in \mathcal{N}_{j}(y(\tau))}{\frac{g(y_j(i))}{f(y_j(i))}}\right]\times\vspace{0.2cm} \\  \hspace{1.6cm}
\hspace{-0.6cm}\displaystyle\left[\prod_{i\in \mathcal{N}_m(y(\tau))}{\hspace{-0.2cm}g(y_m(i))}\prod_{i\in \mathcal{N}_{j}(y(\tau))}{\hspace{-0.2cm}f(y_j(i))}
\times
\vspace{0.2cm}\right.\\\left. \hspace{2cm} \displaystyle
\prod_{k\neq m, j}\;\prod_{i\in \mathcal{N}_k(y(\tau))}{\hspace{-0.3cm}f(y_k(i))}\right]
d\mu(y(\tau)) \vspace{0.2cm} \\
\hspace{-0.6cm}=\displaystyle\sum_{\tau=1}^{\infty}\int_{R_{\tau}} 
\displaystyle\exp\left\{-\Delta S_{m,j}(\tau)\right\}\times \vspace{0.2cm} \\  \hspace{1.6cm}
\hspace{-0.6cm}\displaystyle\left[\prod_{i\in \mathcal{N}_m(y(\tau))}{\hspace{-0.2cm}g(y_m(i))}\prod_{i\in \mathcal{N}_{j}(y(\tau))}{\hspace{-0.2cm}f(y_j(i))}
\times
\vspace{0.2cm}\right.\\\left. \hspace{2cm} \displaystyle
\prod_{k\neq m, j}\;\prod_{i\in \mathcal{N}_k(y(\tau))}{\hspace{-0.3cm}f(y_k(i))}\right]
d\mu(y(\tau)) \vspace{0.2cm} \\
\hspace{-0.6cm}\geq c^{1-\epsilon}\displaystyle\mathbf{P}_m\left(\Delta S_{m,j}(\tau)< -\left(1-\epsilon\right)\log c \;,\; \delta=m|\Gamma\right) \;.
\ena \vspace{0.2cm}
\eeq
Thus,
\beq
\bea{l}
\displaystyle\mathbf{P}_m\left(\Delta S_{m,j}(\tau)< -\left(1-\epsilon\right)\log c \;,\; \delta=m|\Gamma\right) \vspace{0.2cm}\\\hspace{3cm}
=O\left(-c^{\epsilon}\log c\right) \;\;\;\forall j\neq m\;. \vspace{0.2cm}
\ena
\eeq
As a result, by (\ref{eq:Delta_S_m})\vspace{0.2cm}
\beq
\bea{l}
\displaystyle\mathbf{P}_m\left(\Delta S_m(\tau)< -\left(1-\epsilon\right)\log c \;,\; \delta=m|\Gamma\right) \vspace{0.2cm}\\
\leq\displaystyle\sum_{j\neq m}\mathbf{P}_m\left(\Delta S_{m,j}(\tau)< -\left(1-\epsilon\right)\log c \;,\; \delta=m|\Gamma\right) \vspace{0.2cm}\\ \hspace{4cm}
=O\left(-c^{\epsilon}\log c\right) \;. \vspace{0.3cm}
\ena
\eeq
Finally, \vspace{0.2cm}
\beq
\bea{l}
\mathbf{P}_m\left(\Delta S_m(\tau)< -\left(1-\epsilon\right)\log c|\Gamma\right) =O\left(-c^{\epsilon}\log c\right) .
\ena
\eeq
\vspace{0.2cm}
\end{proof}
%
\begin{lemma}
\label{lemma:inequality}
Assume that $K<M$ and \vspace{0.1cm}
\begin{center}
$I^*(M,K)=D(g||f)+\frac{(K-1)D(f||g)}{M-1}$. \vspace{0.1cm}
\end{center}
Define the following function:
\beq
\label{eq:func}
\bea{l}
\displaystyle d(t)\triangleq t\left[D(g||f)+\frac{K\frac{n}{t}-1}{M-1}D(f||g)\right]  \;.
\ena
\eeq
Then, $d(t)$ is monotonically increasing with $t$.
\vspace{0.1cm}
\end{lemma}
\begin{proof}
Note that $I^*(M,K)=D(g||f)+\frac{(K-1)D(f||g)}{M-1}$ implies:
\begin{center}
$\bea{l}
\displaystyle D(g||f)+\frac{(K-1)D(f||g)}{M-1}\geq \frac{K D(f||g)}{M-1} \vspace{0.2cm} \\
\displaystyle \iff D(g||f)\geq\frac{D(f||g)}{M-1} \;.
\ena$
\end{center}
Differentiation $d(t)$ with respect to $t$ yields:
\begin{center}
$\displaystyle \frac{\partial d(t)}{\partial t}=D(g||f)-\frac{D(f||g)}{M-1}\geq 0 \;,$
\end{center}
which completes the proof. \vspace{0.3cm}
\end{proof}
%
For the next lemma we define
\beq
\label{eq:j_star}
j^*(t)\triangleq\arg\min_{j\neq m}N_j(t)
\eeq
as the cell (except cell $m$) which has been observed the lowest number of times up to time~$t$ and
\beq
\label{eq:W_star}
\bea{l}
\displaystyle W_m^*(t)\triangleq\sum_{i=1}^{t}\tilde{\ell}_m(i)\mathbf{1}_m(i)
-\sum_{i=1}^{t}\tilde{\ell}_{j^*(t)}(i)\mathbf{1}_{j^*(t)}(i) \;.
\ena
\eeq
Note that $W_m^*(t)$ is a sum of zero-mean r.v. The following lemma shows that $W_m^*(t)$ is sufficiently small. This result will be used in the proof of Lemma~\ref{lemma:Delta_S} to show that \\ $\mathbf{P}_m\left(\max_{1\leq t\leq n}{\Delta S_m(t)}\geq n\left(I^*(M,K)+\epsilon\right)\;|\;\Gamma\right)\rightarrow 0$ as $n\rightarrow\infty$. \vspace{0.3cm}
\begin{lemma}
\label{lemma:W_star}
For every fixed $\epsilon>0$ there exist $C>0$ and $\gamma>0$ such that
\beq
\label{eq:lemma_W_star}
\bea{l}
\displaystyle \mathbf{P}_m\left(\max_{1\leq t\leq n}{W_m^*(t)}\geq n\epsilon|\Gamma\right)\leq Ce^{-\gamma n}
\ena
\eeq
for all $m=1, ..., M$ and for any policy $\Gamma$.
\vspace{0.1cm}
\end{lemma}
\begin{proof}
Since that $N_m(t), N_{j^*(t)}(t)$ are r.v, we can upper bound (\ref{eq:lemma_W_star}) by summing over any possible values that $N_m(t), N_{j^*(t)}(t)$ can take:
\beq\label{eq:lemma_W_star_bound}
\bea{l}
\displaystyle \mathbf{P}_m\left(\max_{1\leq t\leq n}{W_m^*(t)}\geq n\epsilon|\Gamma\right) \vspace{0.2cm}\\
\leq\displaystyle \sum_{t=1}^{n}\mathbf{P}_m\left(\sum_{r=1}^{t}\tilde{\ell}_m(r)\mathbf{1}_m(r)
-\tilde{\ell}_{j^*(t)}(i)\mathbf{1}_{j^*(t)}(r)\geq n\epsilon|\Gamma\right) \vspace{0.2cm}\\
=\displaystyle \sum_{t=1}^{n}\;\sum_{i=0}^{t}\;\sum_{j=0}^{t}
\vspace{0.2cm}\\ \hspace{0.2cm}
\displaystyle\mathbf{P}_m\left(\sum_{r=1}^{t}\tilde{\ell}_m(r)\mathbf{1}_m(r)
                 +\sum_{r=1}^{t}-\tilde{\ell}_{j^*(t)}(r)\mathbf{1}_{j^*(t)}(r)\geq n\epsilon, \right.\vspace{0.2cm}\\ \hspace{4.5cm}\left.
                                                        N_m(t)=i,N_{j^*(t)}=j |\Gamma\right) \vspace{0.2cm}\\
\leq\displaystyle \sum_{t=1}^{n}\;\sum_{i=0}^{t}\;\sum_{j=0}^{t}
\displaystyle \left[\mathbf{E}_m\left(e^{s(\tilde{\ell}_m(1)-\epsilon/2)}\right)\right]^{i}
\times\vspace{0.2cm}\\ \hspace{2.5cm}
\displaystyle\left[\mathbf{E}_m\left(e^{s(-\tilde{\ell}_{j^*(t)}(1)-\epsilon/2)}\right)\right]^{j}
\times\vspace{0.2cm}\\ \hspace{3.5cm}
\displaystyle \exp\left\{-s\frac{\epsilon}{2}(2n-i-j)\right\} \;,
\ena
\eeq
for all $s>0$.\\
The last inequality follows due to the i.i.d. property of $\ell_k(t)$ across the time series and applying the Chernoff bound for each term in the summation on the RHS of the equality.
Note that we used the fact that the measure of the sample space that satisfies $\left\{\sum_{r=1}^{t}\tilde{\ell}_m(r)\mathbf{1}_m(r)+\sum_{r=1}^{t}-\tilde{\ell}_{j^*(t)}(r)\mathbf{1}_{j^*(t)}(r)\geq n\epsilon,\right.$ $\left. N_m(t)=i,N_{j^*(t)}=j\right\}$ under policy $\Gamma$ is smaller than the measure of the sample space that satisfies $\left\{\sum_{r=1}^{i}\tilde{\ell}_m(r)+\sum_{r=1}^{j}-\tilde{\ell}_{j^*(t)}(r)\geq n\epsilon\right\}$ (which is bounded by the Chernoff bound).
This fact follows since that any selection of $i, j$ observations from cells $m, j^*(t)$ (which have i.i.d distributions), yields the same distribution independent of the time they were taken. In particular, the intersection of the sample space $\left\{\sum_{r=1}^{t}\tilde{\ell}_m(r)\mathbf{1}_m(r)+\sum_{r=1}^{t}-\tilde{\ell}_{j^*(t)}(r)\mathbf{1}_{j^*(t)}(r)\geq n\epsilon,\right\}$ and $\left\{N_m(t)=i,N_{j^*(t)}=j\right\}$ under policy $\Gamma$ further decreases the measure.

Clearly, a moment generating function (MGF) is equal to one at $s=0$. Furthermore, since $\mathbf{E}_m(\tilde{\ell}_m(1)-\epsilon/2)=-\epsilon/2<0$ and $\mathbf{E}_m(-\tilde{\ell}_{j^*(t)}(1)-\epsilon/2)=-\epsilon/2<0$ are strictly negative, differentiating the MGFs of $\tilde{\ell}_m(1)-\epsilon/2$ and $-\tilde{\ell}_{j^*(t)}(1)-\epsilon/2$ with respect to $s$ yields strictly negative derivatives at $s=0$. Hence, there exist $s>0$ and $\gamma'>0$ such that $\mathbf{E}_m\left(e^{s(\tilde{\ell}_m(1)-\epsilon/2)}\right)$, $\mathbf{E}_m\left(e^{s(-\tilde{\ell}_{j^*(t)}(1)-\epsilon/2)}\right)$ and $e^{-s\epsilon/2}$ are strictly less than $e^{-\gamma'}<1$. Since $2n-i-j\geq 0$, there exist $C>0$ and $\gamma>0$, such that summing over $t, i, j$ yields (\ref{eq:lemma_W_star}).
\vspace{0.3cm}
\end{proof}
%
%
%
\begin{lemma}
\label{lemma:Delta_S}
For any fixed $\epsilon>0$,
\beq
\label{eq:lemma_Delta_S}
\bea{l}
\displaystyle \mathbf{P}_m\left(\max_{1\leq t\leq n}{\Delta S_m(t)}\geq n\left(I^*(M,K)+\epsilon\right)\;|\;\Gamma\right)\rightarrow 0 \vspace{0.2cm} \\ \hspace{6cm}
         \;\; \mbox{as} \;\; n\rightarrow\infty \;,
\ena
\eeq
for all $m=1, ..., M$ and for any policy $\Gamma$.
\vspace{0.1cm}
\end{lemma}
\begin{proof}
It should be noted that a polynomial decay of a similar condition under a binary composite hypothesis testing was shown in~\cite[Lemma $5$]{Chernoff_1959_Sequential} using a variation of the Kolmogorov's inequality. Here, we use a different approach to show exponential decay of (\ref{eq:lemma_Delta_S}).
Let
\begin{center}
$\Delta S^*_m(t)\triangleq S_m(t)-S_{j^*(t)}(t)$.
\end{center}
Since $\Delta S_m(t)\leq \Delta S^*_m(t)$ for all $m$ and $t$, we have:
\beq
\label{eq:pr_lemma_Pr_Delta_S}
\bea{l}
\displaystyle \mathbf{P}_m\left(\max_{1\leq t\leq n}{\Delta S_m(t)}\geq n\left(I^*(M,K)+\epsilon\right)|\Gamma\right)\vspace{0.2cm}\\
\leq\displaystyle\mathbf{P}_m\left(\max_{1\leq t\leq n}{\Delta S^*_m(t)}\geq n\left(I^*(M,K)+\epsilon\right)|\Gamma\right)
\ena
\eeq

Next, we consider three cases: \\

\emph{\textbf{Case 1 :} $K=M$:} \\

In this case $I^*(M,K)=D(g||f)+D(f||g)$.
Furthermore, note that $N_j(t)=t$, for all $j$ and $t$. Thus,
\beq
\bea{l}
\Delta S^*_m(t)=W_m^*(t)+t \left(D(g||f)+D(f||g)\right)
\vspace{0.2cm}\\ \hspace{1cm}
\leq W_m^*(t)+n I^*(M,K) \;.
\ena
\eeq
Therefore,
\begin{center}
$\Delta S^*_m(t)\geq n\left(I^*(M,K)+\epsilon\right)$
\end{center}
implies
\begin{center}
$W_m^*(t)\geq n\epsilon$.
\end{center}
Applying Lemma~\ref{lemma:W_star} yields:
\beq\label{eq:pl_Delta_S_1}
\bea{l}
\displaystyle\mathbf{P}_m\left(\max_{1\leq t\leq n}{\Delta S_m(t)}\geq n\left(I^*(M,K)+\epsilon\right)|\Gamma\right)\vspace{0.2cm}\\
\displaystyle\leq\mathbf{P}_m\left(\max_{1\leq t\leq n}{W_m^*(t)}\geq n\epsilon|\Gamma\right) \vspace{0.2cm}\\
\leq Ce^{-\gamma n} \rightarrow 0 \;\;\;\mbox{as}\;\;\; n\rightarrow\infty \;.
\ena
\eeq

\emph{\textbf{Case 2 :} $K<M$ and $I^*(M,K)=\frac{K D(f||g)}{M-1}$:} \\

Note that:
\beq
\bea{l}
\Delta S^*_m(t)=W_m^*(t)+N_m(t)D(g||f)+N_{j^*(t)}(t)D(f||g)
\vspace{0.2cm}\\ \hspace{1cm}
\leq W_m^*(t)+D(f||g)\left[\frac{N_m(t)}{M-1}+N_{j^*(t)}(t)\right]
\ena
\eeq
The last inequality holds since $\frac{K D(f||g)}{M-1}\geq D(g||f)+\frac{(K-1)D(f||g)}{M-1}$ implies $D(g||f)\leq\frac{D(f||g)}{M-1}$.\\
Since that $j^*(t)=\arg\min_{j\neq m}N_j(t)$ and $Kt-N_m(t)$ is the total number of observations taken from $M-1$ cells $j\neq m$, we have:
\beq
\displaystyle N_{j^*(t)}(t)\leq\frac{Kt-N_m(t)}{M-1}\leq\frac{Kn-N_m(t)}{M-1}\;.
\eeq
Hence,
\beq
\bea{l}
\Delta S^*_m(t)
\leq W_m^*(t)+ D(f||g)\frac{K n}{M-1}  \vspace{0.2cm}\\ \hspace{2cm}
=W_m^*(t)+n I^*(M,K) \;.
\ena
\eeq
Therefore,
\begin{center}
$\Delta S^*_m(t)\geq n\left(I^*(M,K)+\epsilon\right)$
\end{center}
implies
\begin{center}
$W_m^*(t)\geq n\epsilon$.
\end{center}
The rest of the proof is similar to Case $1$. \\

\emph{\textbf{Case 3 :} $K<M$ and $I^*(M,K)=D(g||f)+\frac{(K-1)D(f||g)}{M-1}$:} \\

Note that:
\beq
\bea{l}
\displaystyle\Delta S^*_m(t)=W_m^*(t)+N_m(t)D(g||f)+N_{j^*(t)}(t)D(f||g)
\vspace{0.2cm}\\ \hspace{0.5cm}
\displaystyle\leq W_m^*(t)+N_m(t)D(g||f)+\frac{Kn-N_m(t)}{M-1}D(f||g)
\vspace{0.2cm}\\ \hspace{0.5cm}
\displaystyle = W_m^*(t)+N_m(t)\left[ D(g||f)+\frac{K\frac{n}{N_m(t)}-1}{M-1}D(f||g)\right]
\ena
\eeq
Since $0\leq N_m(t)\leq n$, by Lemma~\ref{lemma:inequality} we have:
\beq
\bea{l}
\displaystyle\Delta S^*_m(t)\leq W_m^*(t)+n\left[D(g||f)+\frac{K-1}{M-1}D(f||g)\right]
\vspace{0.2cm}\\ \hspace{2cm}
\displaystyle=W_m^*(t)+nI^*(M,K) \;.
\ena
\eeq
The rest of the proof is similar to Case $1$. Hence, (\ref{eq:lemma_Delta_S}) follows.
\vspace{0.3cm}
\end{proof}
The following theorem shows that in order to achieve a Bayes risk lower than $\frac{-c\log(c)}{I^*(M,K)}$ under any hypothesis, the risk must be of a greater order than $O(-c\log c)$ under some hypothesis. As a result, it provides a lower bound $\sim\frac{-c\log(c)}{I^*(M,K)}$ on the average Bayes risk (\ref{eq:Bayes_risk}): \\
\begin{theorem}
\label{th:lower_bound}
Any policy $\Gamma$ that satisfies $R_j(\Gamma)=O(-c\log c)$ for all $j=1, ..., M$ must satisfy:
\beq
\label{eq:lower_bound}
\displaystyle R_m(\Gamma)\geq -\left(1+o(1)\right)\frac{c\log(c)}{I^*(M,K)} \;.
\eeq
for all $m=1, ..., M$.
\vspace{0.2cm} \\
\end{theorem}
%
\begin{proof}
To show the lower bound on the Bayes risk, we use a similar argument as in~\cite{Chernoff_1959_Sequential}.
For any $\epsilon>0$ let $\displaystyle n_c=-(1-\epsilon)\frac{\log c}{I^*(M,K)+\epsilon}$.
Note that
\beq
\bea{l}
\displaystyle\mathbf{P}_m\left(\tau\leq n_c\;|\;\Gamma\right) \vspace{0.2cm} \\
=\displaystyle\mathbf{P}_m\left(\tau\leq n_c \;,\; \Delta S_m(\tau)\geq -\left(1-\epsilon\right)\log c \;|\;\Gamma\right)   \vspace{0.2cm}\\
+\displaystyle\mathbf{P}_m\left(\tau\leq n_c \;,\; \Delta S_m(\tau)< -\left(1-\epsilon\right)\log c \;|\;\Gamma\right) \vspace{0.2cm} \\
\leq\displaystyle\mathbf{P}_m\left(\max_{t\leq n_c}\Delta S_m(t)\geq -\left(1-\epsilon\right)\log c \;|\;\Gamma\right)   \vspace{0.2cm}\\
+\displaystyle\mathbf{P}_m\left(\Delta S_m(\tau)< -\left(1-\epsilon\right)\log c \;|\;\Gamma\right) \vspace{0.2cm} \;.\\
\ena
\eeq
The first term in the last inequality approaches zero as $c\rightarrow 0$ by Lemma~\ref{lemma:Delta_S}.
Next, note that Lemma~\ref{lemma:Delta_S_O} requires $\alpha_j(\Gamma)=O(-c\log c)$ for all $j=1, ..., M$. Since the theorem requires $R_j(\Gamma)=O(-c\log c)$ for all $j=1, ..., M$ (and recall that $\alpha_j(\Gamma)\leq R_j(\Gamma)$), we can apply Lemma~\ref{lemma:Delta_S_O}. Thus, the second term in the last inequality approaches zero as $c\rightarrow 0$.
As a result, the expected sample size under policy $\Gamma$ satisfies:
\beq
\bea{l}
\displaystyle\mathbf{E}_m(\tau|\Gamma)\geq\sum_{n=n_c+1}^{\infty}n\mathbf{P}_m\left(\tau=n|\Gamma\right)
\vspace{0.2cm}\\\hspace{0.5cm}
\geq n_c\mathbf{P}_m\left(\tau\geq n_c+1|\Gamma\right)
\displaystyle\rightarrow n_c \mbox{\;\;\;as\;\;\;} c\rightarrow 0 \vspace{0.2cm}
\ena
\eeq
Since $\epsilon>0$ can be arbitrarily small we have $\mathbf{E}_m(\tau|\Gamma)\geq-\left(1+o(1)\right)\log(c)/I^*(M,K)$. Hence, $R_m(\Gamma)\geq c\mathbf{E}_m(\tau|\Gamma)\geq -\left(1+o(1)\right)c\log(c)/I^*(M,K)$.
\vspace{0.3cm}
\end{proof}

\subsubsection{Asymptotic Optimality of the DGF policy}
\label{ssec:asymptotic_policy1}
\hspace{0.0cm}\vspace{0.2cm}\\
In this section we show that the DGF policy achieves the lower bound on the Bayes risk (\ref{eq:lower_bound}) as $c\rightarrow 0$. We mainly focus on the more interesting case where $K\geq 2$ cells are observed at a time. The case where $K=1$ is simpler and follows with minor modifications. Below, the proof follows the structure discussed in Section~\ref{sec:performance}.  \vspace{0.2cm}
%
\begin{lemma}
\label{lemma:error_policy1}
Assume that the DGF policy is implemented. Then, the error probability is upper bounded by:
\beq
\label{eq:Pe_bound_policy1}
P_e\leq (M-1)c \;.
\eeq
\vspace{0.1cm}
\end{lemma}
%
\begin{proof}
Let $\alpha_{m,j}=\mathbf{P}_m(\delta=j)$ for all $j\neq m$. Thus, $\alpha_m=\sum_{j\neq m}\alpha_{m,j}$.
Note that accepting $H_j$ (i.e., $\Delta S_j(n)\geq-\log c$) implies $\Delta S_{j,m}\geq-\log c$.
By changing the measure, as in~\cite[Lemma $3$]{Chernoff_1959_Sequential}, we can show that for all $j\neq m$ the following holds:
\beq
\bea{l}
\alpha_{m,j}=\mathbf{P}_m\left(\delta=j\right)\vspace{0.2cm} \\
=\displaystyle\mathbf{P}_m\left(\Delta S_j(\tau)\geq-\log c\right)
\leq\mathbf{P}_m\left(\Delta S_{j,m}(\tau)\geq-\log c\right)  \vspace{0.2cm} \\
\leq c\displaystyle\mathbf{P}_j\left(\Delta S_{j,m}(\tau)\geq -\log c\right)\leq c \;.
\ena
\eeq
Finally,
\begin{center}
$\displaystyle\alpha_m=\sum_{j\neq m}\alpha_{m,j}\leq (M-1)c$\;.
\end{center}
Hence, (\ref{eq:Pe_bound_policy1}) follows.
\vspace{0.3cm}
\end{proof}

\begin{lemma}
\label{lemma:S_j_S_m_N_j}
Fix $0<q<1$. Then, there exist $C>0$ and $\gamma>0$ such that
\beq
\label{eq:lemma:S_j_S_m_N_j}
\displaystyle\mathbf{P}_m\left(S_j(n)\geq S_m(n), N_j(n)\geq qn\right)\leq C e^{-\gamma n} \;,
\eeq
and
\beq
\label{eq:lemma:S_j_S_m_N_m}
\displaystyle\mathbf{P}_m\left(S_j(n)\geq S_m(n), N_m(n)\geq qn\right)\leq C e^{-\gamma n} \;,
\eeq
hold under any policy for $m=1, 2, ..., M$ and $j\neq m$.
\vspace{0.2cm} \\
\end{lemma}
\begin{proof}
We prove (\ref{eq:lemma:S_j_S_m_N_j}). Proving (\ref{eq:lemma:S_j_S_m_N_m}) applies with minor modifications.
Note that we can develop (\ref{eq:lemma:S_j_S_m_N_j}) by summing over any possible values that $N_j(n), N_m(n)$ can take (i.e., $N_j(n)=\lceil qn\rceil, \lceil qn\rceil+1, ...n$, and $N_m(n)=0, ..., n$). Similar to (\ref{eq:lemma_W_star_bound}), applying the Chernoff bound and using the i.i.d. property of $\ell_{j}(t), \ell_{m}(t)$ across time yield:
\beq
\bea{l}
\displaystyle\mathbf{P}_m\left(S_j(n)\geq S_m(n), N_j(n)\geq qn\right)\vspace{0.2cm}\\
\leq\displaystyle\sum_{r=\lceil qn\rceil}^{n}\;\sum_{k=0}^{n}\mathbf{P}_m\left(\sum_{i=1}^{r}\ell_{j}(i)
                                                              +\sum_{i=1}^{k}-\ell_{m}(i)\geq 0\right)
\vspace{0.2cm}\\
\leq\displaystyle\sum_{r=\lceil qn\rceil}^{n}\;\sum_{k=0}^{n}\left[\mathbf{E}_m\left(e^{s\ell_j(1)}\right)\right]^{r}
    \left[\mathbf{E}_m\left(e^{s(-\ell_m(1))}\right)\right]^{k}
\ena
\eeq
for all $s>0$.

Note that a moment generating function (MGF) is equal to one at $s=0$. Furthermore, since $\mathbf{E}_m(\ell_j(1))=-D(f||g)<0$ and $\mathbf{E}_m(-\ell_m(1))=-D(g||f)<0$ are strictly negative, differentiating the MGFs of $\ell_j(1), \ell_m(1)$ with respect to $s$ yields strictly negative derivatives at $s=0$. As a result, there exist $s>0$ and $\gamma_1>0$ such that $\mathbf{E}_m\left(e^{s\ell_j(1)}\right)$, and $\mathbf{E}_m\left(e^{s(-\ell_m(1))}\right)$ are strictly less than $e^{-\gamma_1}<1$. Hence, there exist $C>0$ and $\gamma=\gamma_1\cdot q>0$ such that
\beq
\bea{l}
\displaystyle\mathbf{P}_m\left(S_j(n)-S_m(n)\geq 0, N_j(n)\geq qn\right) \vspace{0.2cm}\\
\leq\displaystyle\sum_{r=\lceil qn\rceil}^{n} e^{-\gamma_1 r}  \sum_{k=0}^{n} e^{-\gamma_1 k}
\leq\displaystyle C e^{-\gamma n} \;.
\ena
\eeq
\vspace{0.1cm}
\end{proof}
%
%
\vspace{0.1cm}
For the following definition, recall that $S_m(n)$ is a random walk with positive expected increment $\mathbf{E}_m(\ell_m(n))=D(g||f)>0$, while $S_j(n)$, for $j\neq m$ is a random walk with negative expected increment $\mathbf{E}_m(\ell_j(n))=-D(f||g)<0$. As a result, ultimately, the sample path of $S_m(n)$ will dominate those of $S_j(n)$, $\forall j\neq m$, when $n$ (and also the number of samples taken from cells $m$ or $j$) is sufficiently large. Below, we define a random time~$\tau_1$, which is the last passage time where the sample path of $S_m(n)$ will dominate those of $S_j(n)$ for all $n\geq\tau_1$, i.e., $\tau_1$ is the last passage time in which $S_m(n)$ crosses $S_j(n)$. It should be noted that $\tau_1$ is not a stopping time (note that $\tau_1$ depends on the future by definition) and the decision maker does not know whether $\tau_1$ has arrived. In Lemma~\ref{lemma:tau_1_policy1} below we show that $\tau_1$ is sufficiently small with high probability. We will use this result later to upper bound the actual stopping time~$\tau$ under DGF.
\vspace{0.1cm}
\begin{definition}
$\tau_1$ is the smallest integer such that $S_m(n)>S_j(n)$ for all $j\neq m$ for all $n\geq\tau_1$. \vspace{0.2cm}
\end{definition}

\begin{remark}
In the following lemmas, when we say that the \emph{DGF policy is implemented indefinitely} we mean that DGF probes the cells indefinitely according to its selection rule, while the stopping rule is disregarded. \vspace{0.2cm}
\end{remark}

\begin{lemma}
\label{lemma:tau_1_policy1}
Assume that the DGF policy is implemented indefinitely.
Then, there exist $C>0$ and $\gamma>0$ such that
\beq
\label{eq:lemma_Pr_tau_1}
\mathbf{P}_m\left(\tau_1>n\right)\leq Ce^{-\gamma n}  \;,
\eeq
for $m=1, 2, ..., M$.
\vspace{0.2cm} \\
\end{lemma}
%
\begin{proof}
We focus on the case where $M>2$. The case of $M=2$ is simpler and follows with minor modifications.
Note that:
\beq\label{eq:tau1}
\bea{l}
\displaystyle\mathbf{P}_m\left(\tau_1>n\right)\leq\mathbf{P}_m\left(\max_{j\neq m}\;\sup_{t\geq n}\;\left(S_j(t)-S_m(t)\right)\geq 0 \right) \vspace{0.2cm} \\ \hspace{2cm}
\leq\displaystyle\sum_{j\neq m}\;\sum_{t=n}^{\infty}\mathbf{P}_m\left(S_j(t)\geq S_m(t)\right)\;.
\ena
\eeq
Therefore, it suffices to show that there exist $C>0$ and $\gamma>0$ such that $\mathbf{P}_m\left(S_j(n)\geq S_m(n)\right)\leq Ce^{-\gamma n}$. \\

\emph{\textbf{Step 1:} Bounding each term in the summation on the RHS of (\ref{eq:tau1}):}

Let
\begin{center}
$\displaystyle\rho=\frac{1}{16(M-2)}$\;.
\end{center}
Note that $0<\rho\leq 1/16$. \\
Thus,
\beq
\bea{l}
\label{eq:l_tau_1_policy1_Sj_geq_Sm}
\mathbf{P}_m\left(S_j(n)\geq S_m(n)\right) \vspace{0.2cm} \\ 
\leq\mathbf{P}_m\left(S_j(n)\geq S_m(n), N_j(n)<\rho n , N_m(n)<\rho n\right)
\vspace{0.2cm} \\ 
+\mathbf{P}_m\left(S_j(n)\geq S_m(n), N_j(n)\geq \rho n\right) \vspace{0.2cm} \\ 
+\mathbf{P}_m\left(S_j(n)\geq S_m(n), N_m(n)\geq \rho n\right)
\ena
\eeq
By Lemma 6, there exist $\gamma_1>0$ and $D>0$ such that the second and the third terms on the RHS are upper bounded by $D e^{-\gamma_1 n}$. In the case of $K=M$ the first term on the RHS equals zero (since $N_j(n)=N_m(n)=n$ surely). Hence, it remains to show that the first term on the RHS decreases exponentially with $n$ for $K<M$.
Note that the event $(N_j(n)<\rho n, N_m(n)<\rho n)$ implies that at least $\tilde{n}=n-N_j(n)-N_m(n)\geq n\left(1-2\rho\right)$ times cells $j,m$ are not observed. Let $\widetilde{N}_r(n)$ be the number of times when cell $r\neq j, m$ has been observed and cells $j,m$ have not been observed up to time $n$. We refer to each such time as \textbf{\emph{$r_{\neq j,m}$-probing time}}. There exists a cell $r\neq j,m$ such that $\widetilde{N}_r(n)\geq\frac{\tilde{n}}{M-2}=\frac{n(1-2\rho)}{M-2}$. Hence, (\ref{eq:l_tau_1_policy1_Sj_geq_Sm}) can be upper bounded by:
\beq\label{eq:Pm_Sj_Sm}
\bea{l}
\mathbf{P}_m\left(S_j(n)\geq S_m(n)\right) \vspace{0.2cm} \\ 
\displaystyle\leq\sum_{r\neq j,m}\mathbf{P}_m\left(\tilde{N}_r(n)>\frac{n(1-2\rho)}{M-2}, \right. \vspace{0.2cm} \\ \hspace{3cm}
\displaystyle\left. N_j(n)<\rho n, N_m(n)<\rho n\right)
\vspace{0.2cm} \\ \hspace{3cm}
+2D e^{-\gamma_1 n}
\ena
\eeq
It remains to show that each term in the summation on the RHS of (\ref{eq:Pm_Sj_Sm}) decreases exponentially with $n$. \\

\emph{\textbf{Step 2:} Bounding each term in the summation on the RHS of (\ref{eq:Pm_Sj_Sm}):}

Let $\tilde{t}^r_1, \tilde{t}^r_2, ..., \tilde{t}^r_{\tilde{N}_r(n)}$ be the \emph{$r_{\neq j,m}$-probing time} indices and let
\begin{center}
$\displaystyle\zeta\triangleq\frac{1-2\rho}{2(M-2)} \;.$
\end{center}
Note that:
\begin{itemize}
    \item At every \emph{$r_{\neq j,m}$-probing time}, $S_j(n)\leq S_r(n)$ or $S_m(n)\leq S_r(n)$ must occur (otherwise, if $S_j(n)>S_r(n)$ and $S_m(n)>S_r(n)$ then $j$ or $m$ are observed).
  \item In particular, the event $\widetilde{N}_r(n)>\frac{n(1-2\rho)}{M-2}$ implies that at time $\tilde{t}^r_{\zeta n}$ the following holds: $S_j(\tilde{t}^r_{\zeta n})\leq S_r(\tilde{t}^r_{\zeta n})$ or $S_m(\tilde{t}^r_{\zeta n})\leq S_r(\tilde{t}^r_{\zeta n})$ must occur.
  \item Since $N_r\left(t\right)$ is the total number of observations taken from cell $r$ up to time $t$ (during both $r_{\neq j,m}$-probing times and all other times when cell $r$ was probed), then $N_r\left(\tilde{t}^r_{\zeta n}\right)\geq\zeta n$.
\end{itemize}

Therefore, using the i.i.d. property of the LLRs across time we have\footnote{For the ease of presentation, throughout the proof we assume that $\zeta n$, $\rho n$ are integers. This assumption does not affect the exponential decay of the Chernoff bound but only the exact value of $C>0$ in (\ref{eq:lemma_Pr_tau_1}) (since $\alpha n-1\leq\lfloor\alpha n\rfloor\leq\lceil\alpha n\rceil\leq\alpha n+1$ holds for all $\alpha\geq 0$ for all $n=0, 1, ... $).}:
\beq\label{eq:tilde_Nr}
\bea{l}
\mathbf{P}_m\left(\tilde{N}_r(n)>\frac{n(1-2\rho)}{M-2}, 
N_j(n)<\rho n, N_m(n)<\rho n\right) \vspace{0.2cm} \\
\leq\displaystyle\sum_{N'_r=\zeta n}^{n}\mathbf{P}_m\left(\inf_{n'\leq \rho n} \sum_{i=1}^{n'}\ell_j(i)\leq \sum_{i=1}^{N'_r}\ell_r(i)\right) \vspace{0.2cm} \\ \hspace{0.5cm}
+\displaystyle\sum_{N'_r=\zeta n}^{n}\mathbf{P}_m\left(\inf_{n'\leq \rho n} \sum_{i=1}^{n'}\ell_m(i)\leq \sum_{i=1}^{N'_r}\ell_r(i)\right) \vspace{0.2cm} \\
\leq\displaystyle\sum_{N'_r=\zeta n}^{n}\;\sum_{n'=0}^{\rho n}
\mathbf{P}_m\left(\sum_{i=1}^{n'}\ell_j(i)\leq \sum_{i=1}^{N'_r}\ell_r(i)\right) \vspace{0.2cm} \\ \hspace{0.5cm}
+\displaystyle\sum_{N'_r=\zeta n}^{n}\;\sum_{n'=0}^{\rho n}
\mathbf{P}_m\left(\sum_{i=1}^{n'}\ell_m(i)\leq \sum_{i=1}^{N'_r}\ell_r(i)\right) \vspace{0.2cm} \\
=\displaystyle\sum_{q=0}^{n-\zeta n}\;\sum_{n'=0}^{\rho n}
\mathbf{P}_m\left(\sum_{i=1}^{n'}\ell_j(i)\leq \sum_{i=1}^{\zeta n+q}\ell_r(i)\right) \vspace{0.2cm} \\ \hspace{0.5cm}
+\displaystyle\sum_{q=0}^{n-\zeta n}\;\sum_{n'=0}^{\rho n}
\mathbf{P}_m\left(\sum_{i=1}^{n'}\ell_m(i)\leq \sum_{i=1}^{\zeta n+q}\ell_r(i)\right) \vspace{0.2cm} \\
\ena
\eeq

\emph{\textbf{Step 3:} Bounding the first term on the RHS of (\ref{eq:tilde_Nr}):}

Note that
\beq
\bea{l}
\displaystyle \sum_{i=1}^{\zeta n+q}\ell_r(i)+\sum_{i=1}^{n'}-\ell_j(i) \vspace{0.2cm}\\ 
=\displaystyle \sum_{i=1}^{\zeta n+q}\tilde{\ell}_r(i)+\sum_{i=1}^{n'}-\tilde{\ell}_j(i) 
                    -D(f||g)\left(\zeta n+q-n'\right)
                    \;.
\ena
\eeq
and
\begin{center}
$\bea{l}
\displaystyle\zeta n+q-n'\geq \zeta n+q-n'-2\left(\rho n-n'\right)\vspace{0.1cm}\\
\displaystyle=n\left(\zeta-2\rho\right)+q+n'\geq\frac{1}{4(M-2)}n+q+n'\vspace{0.1cm}\\
\displaystyle\geq\frac{1}{4(M-2)}(n+q+n')\;,
\vspace{0.1cm}\\
\ena$
\end{center}
for all $n'\leq\rho n$. \\
Therefore,
\beq
\bea{l}
\displaystyle\sum_{i=1}^{\zeta n+q}\ell_r(i)+\sum_{i=1}^{n'}-\ell_j(i)\geq 0
\ena
\eeq
implies
\beq
\bea{l}
\displaystyle\sum_{i=1}^{\zeta n+q}\tilde{\ell}_r(i)+\sum_{i=1}^{n'}-\tilde{\ell}_j(i)\geq C_1\left(n+q+n'\right)\;.
\ena
\eeq
where $C_1=\frac{D(f||g)}{4(M-2)}>0$. \\
Applying the Chernoff bound and using the i.i.d. property of $\ell_r(t), \ell_j(t)$ across the time series yield:
\beq\label{eq:Chernoff}
\bea{l}
\displaystyle\mathbf{P}_m\left(\sum_{i=1}^{n'}\ell_j(i)\leq \sum_{i=1}^{\zeta n+q}\ell_r(i)\right) \vspace{0.2cm}\\
\displaystyle\leq\mathbf{P}_m\left(\sum_{i=1}^{\zeta n+q}\tilde{\ell}_r(i)+\sum_{i=1}^{n'}-\tilde{\ell}_j(i)\geq C_1\left(n+q+n'\right)\right) \vspace{0.2cm}\\
\leq\displaystyle\left[\mathbf{E}_m\left(e^{s\tilde{\ell}_r(1)}\right)\right]^{\zeta n+q}
    \left[\mathbf{E}_m\left(e^{s(-\tilde{\ell}_j(1))}\right)\right]^{n'} \times\vspace{0.2cm}\\ \hspace{3cm} \displaystyle e^{-sC_1\left(n+q+n'\right)} \vspace{0.2cm}\\
=\displaystyle\left[\mathbf{E}_m\left(e^{s\left(\tilde{\ell}_r(1)-C_1\right)}\right)\right]^{\zeta n+q} \times\vspace{0.2cm}\\ \hspace{1cm}
    \displaystyle\left[\mathbf{E}_m\left(e^{s\left(-\tilde{\ell}_j(1)-C_1\right)}\right)\right]^{n'} \times\vspace{0.2cm}\\ \hspace{3cm}
    \displaystyle e^{-sC_1\left(n-\zeta n\right)}
\;.
\ena
\eeq
for all $s>0$.\\
Since $\mathbf{E}_m(\tilde{\ell}_r(1)-C_1)=-C_1<0$ and $\mathbf{E}_m(-\tilde{\ell}_j(1)-C_1)=-C_1<0$ are strictly negative, by applying a similar argument as at the end of the proof of Lemma \ref{lemma:S_j_S_m_N_j}, there exist $s>0$ and $\gamma_2>0$ such that $\mathbf{E}_m\left(e^{(s\tilde{\ell}_r(1)-C_1)}\right)$, $\mathbf{E}_m\left(e^{s(-\tilde{\ell}_j(1)-C_1)}\right)$ and $e^{-sC_1}$ are strictly less than $e^{-\gamma_2}<1$.
Hence,
\beq\label{eq:Chernoff}
\bea{l}
\displaystyle\mathbf{P}_m\left(\sum_{i=1}^{n'}\ell_j(i)\leq \sum_{i=1}^{\zeta n+q}\ell_r(i)\right)
    \leq\displaystyle e^{-\gamma_2\left(n+q+n'\right)}
\;.
\ena
\eeq
and
\beq
\bea{l}
\displaystyle\sum_{q=0}^{n-\zeta n}\;\sum_{n'=0}^{\rho n}
\mathbf{P}_m\left(\sum_{i=1}^{n'}\ell_j(i)\leq \sum_{i=1}^{\zeta n+q}\ell_r(i)\right) \vspace{0.2cm} \\ \hspace{0.2cm}
\displaystyle\leq e^{-\gamma_2 n} \sum_{q=0}^{n-\zeta n} e^{-\gamma_2 q}
\sum_{n'=0}^{\rho n} e^{-\gamma_2 n'}\leq C_2 e^{-\gamma_2 n} \;,
\ena
\eeq
where $C_2=\left(1-e^{-\gamma_2}\right)^{-2}$.\\

\emph{\textbf{Step 4:} Bounding the second term on the RHS of (\ref{eq:tilde_Nr}):}

Applying the Chernoff bound and using the i.i.d. property of $\ell_r(t), \ell_m(t)$ across the time series yield:
\beq\label{eq:Chernoff}
\bea{l}
\displaystyle\mathbf{P}_m\left(\sum_{i=1}^{\zeta n+q}\ell_r(i)+\sum_{i=1}^{n'}-\ell_m(i)\geq 0\right) \vspace{0.2cm}\\
\leq\displaystyle\left[\mathbf{E}_m\left(e^{s\ell_r(1)}\right)\right]^{\zeta n+q}
    \left[\mathbf{E}_m\left(e^{s(-\ell_m(1))}\right)\right]^{n'}
\;.
\ena
\eeq
for all $s>0$.\\

Since $\mathbf{E}_m(\ell_r(1))=-D(f||g)<0$ and $\mathbf{E}_m(-\ell_m(1))=-D(g||f)<0$ are strictly negative, there exist $s>0$ and $\gamma'_3>0$ such that $\mathbf{E}_m\left(e^{s\ell_r(1)}\right)$, $\mathbf{E}_m\left(e^{s(-\ell_m(1))}\right)$ are strictly less than $e^{-\gamma'_3}<1$. Hence,
\beq\label{eq:Chernoff}
\bea{l}
\displaystyle\mathbf{P}_m\left(\sum_{i=1}^{\zeta n+q}\ell_r(i)+\sum_{i=1}^{n'}-\ell_m(i)\geq 0\right)
\leq e^{-\gamma'_3(\zeta n+q+n')}
\;.
\ena
\eeq
Finally, there exists $\gamma_3=\zeta\gamma'_3>0$ such that
\beq
\bea{l}
\displaystyle\sum_{q=0}^{n-\zeta n}\;\sum_{n'=0}^{\rho n}
\mathbf{P}_m\left(\sum_{i=1}^{n'}\ell_m(i)\leq \sum_{i=1}^{\zeta n+q}\ell_r(i)\right) \vspace{0.2cm} \\ \hspace{0.2cm}
\displaystyle\leq e^{-\gamma_3 n} \sum_{q=0}^{n-\zeta n} e^{-\gamma_3 q/\zeta}
\sum_{n'=0}^{\rho n} e^{-\gamma_3 n'/\zeta}\leq \frac{e^{-\gamma_3 n}}{\left(1-e^{-\gamma_3/\zeta}\right)^2} \;,
\ena
\eeq
which completes the proof.
\vspace{0.2cm}
\end{proof}
It should be noted that differing from~\cite{Nitinawarat_2013_Controlled}, where only a polynomial decay of $\mathbf{P}_m(\tau_1>n)$ was shown to handle indistinguishable hypotheses under some (but not all) actions when applying the extended randomized Chernoff test, Lemma~\ref{lemma:tau_1_policy1} shows exponential decay of $\mathbf{P}_m(\tau_1>n)$ under DGF. \vspace{0.1cm}

For the next lemmas we define
\beq\label{eq:D_prime_fg}
\displaystyle D'(f||g)\triangleq(K-1)D(f||g)/(M-1) \;.  \vspace{0.2cm}
\eeq
In what follows we define the second random time~$\tau_2\geq\tau_1$. $\tau_2$ can be viewed as the time where sufficient information has been gathered to distinguish hypothesis $m$ from at least one false hypothesis $j\neq m$. We point out that $\tau_2$ is not a stopping time. However, it serves us later in upper bounding the stopping time~$\tau$ under DGF. Lemma~\ref{lemma:tau_2_policy1} shows that in the asymptotic regime the total time between $\tau_1$ and $\tau_2$ approaches $\sim-\log c/I^*(M,K)$. \vspace{0.2cm}
\begin{definition}
$\tau_2$ is defined as follows:
\begin{enumerate}
  \item If $K=M$, $\tau_2$ denotes the smallest integer such that $\sum_{i=\tau_1+1}^{n}\ell_m(i)\mathbf{1}_m(i)\geq-\frac{D(g||f)}{I^*(M,K)}\log c$ and $\sum_{i=\tau_1+1}^{n}\ell_{j_n}(i)\mathbf{1}_{j_n}(i)\leq\frac{D(f||g)}{I^*(M,K)}\log c$ for some $j_n\neq m$ for all $n\geq\tau_2\geq\tau_1$.\vspace{0.2cm}
  \item If $K<M$ and $I^*(M,K)=KD(f||g)/(M-1)$, $\tau_2$ denotes the smallest integer such that $\sum_{i=\tau_1+1}^{n}\ell_{j_n}(i)\mathbf{1}_{j_n}(i)\leq\log c$ for some $j_n\neq m$ for all $n\geq\tau_2\geq\tau_1$.  \vspace{0.2cm}
  \item If $K<M$ and $I^*(M,K)=D(g||f)+(K-1)D(f||g)/(M-1)$, $\tau_2$ denotes the smallest integer such that $\sum_{i=\tau_1+1}^{n}\ell_m(i)\mathbf{1}_m(i)\geq-\frac{D(g||f)}{I^*(M,K)}\log c$ and $\sum_{i=\tau_1+1}^{n}\ell_{j_n}(i)\mathbf{1}_{j_n}(i)\leq\frac{D'(f||g)}{I^*(M,K)}\log c$ for some $j_n\neq m$ for all $n\geq\tau_2\geq\tau_1$.\vspace{0.2cm}
\end{enumerate}
\end{definition}
\begin{definition}
$n_2\triangleq\tau_2-\tau_1$ denotes the total amount of time between $\tau_1$ and $\tau_2$. \vspace{0.2cm}
\end{definition}
\begin{lemma}
\label{lemma:tau_2_policy1}
Assume that the DGF policy is implemented indefinitely.
Then, for every fixed $\epsilon>0$ there exist $C>0$ and $\gamma>0$ such that
\beq
\mathbf{P}_m\left(n_2>n\right)\leq C e^{-\gamma n} \;\;\;\; \forall n>-(1+\epsilon)\log c/I^*(M,K)\;,
\eeq
for all $m=1, 2, ..., M$.
\vspace{0.2cm}
\end{lemma}
%
\begin{proof}
We prove the lemma for three cases: \vspace{0.1cm}\\

\emph{\textbf{Case 1 :} $K=M$:} \\
In this case $I^*(M,K)=D(g||f)+D(f||g)$ and $\mathbf{1}_k(t)=1$ for all $k, t$. Let $\tau_2^m$ and $\tau_2^j$ for $j\neq m$ be the smallest integers such that $\sum_{i=\tau_1+1}^{n}\ell_m(i)\geq-\frac{D(g||f)}{I^*(M,K)}\log c$ for all $n>\tau_2^m$ and $\sum_{i=\tau_1+1}^{n}\ell_j(i)\leq\frac{D(f||g)}{I^*(M,K)}\log c$ for all $n>\tau_2^j$ for $j\neq m$, respectively. Similarly, $n_2^k$ denotes the total amount of time between $\tau_1$ and $\tau_2^k$. Clearly, $n_2\leq\max_k(n_2^k)$. As a result,
\beq
\bea{l}
\displaystyle\mathbf{P}_m\left(n_2>n\right)
\displaystyle\leq\sum_{k=1}^{M}\mathbf{P}_m\left(n_2^k>n\right) \;.
\ena
\eeq
Thus it remains to show that $\mathbf{P}_m\left(n_2^k>n\right)$ decreases exponentially with $n$.
Next, we prove the lemma for cell $m$. The proof for cell $j\neq m$ follows with minor modifications.
Let $\epsilon_1=D(g||f)\epsilon/(1+\epsilon)>0$. Thus,
\beq
\bea{l}
\displaystyle\sum_{i=\tau_1+1}^{t+\tau_1}\ell_m(i)+\frac{D(g||f)}{I^*(M,K)}\log c \vspace{0.1cm}\\
=\displaystyle\sum_{i=\tau_1+1}^{t+\tau_1}\tilde{\ell}_m(i)+t D(g||f)+\frac{D(g||f)}{I^*(M,K)}\log c
\vspace{0.1cm}\\ 
\geq\displaystyle\sum_{i=\tau_1+1}^{t+\tau_1}\tilde{\ell}_m(i)+t\epsilon_1
\;,
\ena
\eeq
for all $t\geq n>-(1+\epsilon)\log c/I^*(M,K)$. \\
As a result,
\beq
\bea{l}
\displaystyle\sum_{i=\tau_1+1}^{t+\tau_1}\ell_m(i)\leq-\frac{D(g||f)}{I^*(M,K)}\log c \;.
\ena
\eeq
implies
\beq
\bea{l}
\displaystyle\sum_{i=\tau_1+1}^{t+\tau_1}\tilde{\ell}_m(i)\leq-t\epsilon_1 \;.
\ena
\eeq
Hence, for any $\epsilon>0$ there exists $\epsilon_1>0$ such that
\beq
\bea{l}
\displaystyle\mathbf{P}_m\left(n_2^m>n\right)\vspace{0.1cm}\\
\displaystyle\leq\mathbf{P}_m\left(\inf_{t\geq n}\sum_{i=\tau_1+1}^{t+\tau_1}\ell_m(i)\leq-\frac{D(g||f)}{I^*(M,K)}\log c\right)\vspace{0.1cm}\\
\displaystyle\leq\sum_{t=n}^{\infty}\mathbf{P}_m\left(\sum_{i=\tau_1+1}^{t+\tau_1}\ell_m(i)\leq-\frac{D(g||f)}{I^*(M,K)}\log c\right)\vspace{0.1cm}\\
\displaystyle\leq\sum_{t=n}^{\infty}\mathbf{P}_m\left(\sum_{i=\tau_1+1}^{t+\tau_1}-\tilde{\ell}_m(i)\geq t\epsilon_1\right) \;,
\ena
\eeq
for all $t\geq n>-(1+\epsilon)\log c/I^*(M,K)$.\\
By applying the Chernoff bound, it can be shown that there exists $\gamma_1>0$ such that: $\mathbf{P}_m\left(\sum_{i=\tau_1+1}^{t+\tau_1}-\tilde{\ell}_m(i)\geq t\epsilon_1\right)\leq e^{-\gamma_1 t}$ for all $t\geq n>-(1+\epsilon)\log c/I^*(M,K)$. Hence, there exist $C_1>0$ and $\gamma_1>0$ such that $\mathbf{P}_m\left(n_2^m>n\right)\leq C_1e^{-\gamma_1 n}$ for all $n>-(1+\epsilon)\log c/I^*(M,K)$.  \\

\emph{\textbf{Case 2 :} $K<M$ and $I^*(M,K)=\frac{K D(f||g)}{M-1}$:} \\
In this case, the cell with the highest index is not observed for all $n$. As a result, cell $m$ is not observed for all $n\geq\tau_1$ since $S_m(n)>S_j(n)$ for all $j\neq m$ for all $n\geq\tau_1$.
Let
$N'_j(\tau_1+t)\triangleq\sum_{i=\tau_1+1}^{\tau_1+t}\mathbf{1}_j(i)$.
Let $j^*(\tau_1+t)=\arg\;\max_{j\neq m}N'_j(\tau_1+t)$ be the cell index that was observed the largest number of times since $\tau_1$ has occurred up to time~$\tau_1+t$. Note that if $\sum_{i=\tau_1+1}^{\tau_1+t}\ell_{j^*(\tau_1+t)}(i)\mathbf{1}_{j^*(\tau_1+t)}(i)\leq\log c$ for all $t\geq n$, then $n_2\leq n$. Hence,
\beq
\bea{l}
\displaystyle\mathbf{P}_m\left(n_2>n\right)\vspace{0.1cm}\\\hspace{0.5cm}
\displaystyle\leq\mathbf{P}_m\left(\sup_{t\geq n}\sum_{i=\tau_1+1}^{\tau_1+t}\ell_{j^*(\tau_1+t)}(i)\mathbf{1}_{j^*(\tau_1+t)}(i)\geq\log c\right) \;.
\ena
\eeq
Note that $N'_{j^*(\tau_1+t)}(\tau_1+t)\geq \frac{Kt}{M-1}$ (since $Kt$ observations are taken from $M-1$ cells). Thus, for any $\epsilon>0$ there exists $\epsilon_1>0$ such that:
\beq
\label{eq:lemma:tau_2_policy1_1}
\bea{l}
\displaystyle\sum_{i=\tau_1+1}^{\tau_1+t}\ell_{j^*(\tau_1+t)}(i)\mathbf{1}_{j^*(\tau_1+t)}(i)-\log c \vspace{0.1cm}\\
=\displaystyle\sum_{i=\tau_1+1}^{\tau_1+t}\tilde{\ell}_{j^*(\tau_1+t)}(i)\mathbf{1}_{j^*(\tau_1+t)}(i)
\vspace{0.1cm}\\\hspace{1.5cm}
-N'_{j^*(\tau_1+t)}(\tau_1+t)D(f||g)-\log c
\vspace{0.1cm}\\
\leq\displaystyle\sum_{i=\tau_1+1}^{\tau_1+t}\tilde{\ell}_{j^*(\tau_1+t)}(i)\mathbf{1}_{j^*(\tau_1+t)}(i) \vspace{0.1cm}\\ \hspace{2cm}
\displaystyle-\frac{tD(f||g)K}{M-1}\left(1-\frac{-(M-1)\log c}{tKD(f||g)}\right)
\vspace{0.1cm}\\
\leq\displaystyle\sum_{i=\tau_1+1}^{\tau_1+t}\tilde{\ell}_{j^*(\tau_1+t)}(i)\mathbf{1}_{j^*(\tau_1+t)}(i)-t\epsilon_1
\;,
\ena
\eeq
for all $t\geq n>-(1+\epsilon)\log c/I^*(M,K)=-(1+\epsilon)(M-1)\log c/(KD(f||g)$. \\
The rest of the proof follows by by applying the Chernoff bound.  \\

\emph{\textbf{Case 3 :} $K<M$ and $I^*(M,K)=D(g||f)+\frac{(K-1)D(f||g)}{M-1}$:} \\
In this case, the cell with the highest index is observed for all $n$. As a result, cell $m$ is observed for all $n\geq\tau_1$ since $S_m(n)>S_j(n)$ for all $j\neq m$ for all $n\geq\tau_1$. Therefore, a similar argument as in Case $1$ applies to cell $m$.
Next, we focus on cell $j^*(\tau_1+t)\neq m$ as in Case $2$. Note that in this case $N_{j^*(\tau_1+t)}(\tau_1+t)\geq \frac{(K-1)t}{M-1}$ (since $t$ observations are taken from cell $m$ and $(K-1)t$ observations are taken from $M-1$ cells (for $j\neq m$)). Thus,  (\ref{eq:lemma:tau_2_policy1_1}) can be developed for this case as well with minor modifications.
\vspace{0.3cm}
\end{proof}
In what follows we define the dynamic range of the false hypotheses in terms of their sum LLRs. Note that the dynamic range at time~$\tau_2$ (which is the time where sufficient information has been gathered to distinguish $H_m$ from at least one false hypothesis) can be viewed as the total amount of information remains to gather in order to distinguish $H_m$ from all the false hypotheses. Lemma~\ref{lemma:DR_tau1_policy1} shows that the dynamic range at time~$\tau_2$ is sufficiently small. \vspace{0.3cm}
\begin{definition}
The dynamic range of the false hypotheses at time~$t$ is defined as follows:
\beq
\displaystyle \DR(t)\triangleq\max_{j\neq m}S_j(t)-\min_{j\neq m}S_j(t) \;. \vspace{0.2cm}
\eeq
\end{definition}
%
\begin{lemma}
\label{lemma:DR_tau1_policy1}
Assume that the DGF policy is implemented indefinitely.
Then, for every fixed $\epsilon_1>0, \epsilon_2>0$ there exist $C>0$ and $\gamma>0$ such that
\beq
\bea{l}
\displaystyle\mathbf{P}_m\left(\DR(\tau_2)>\epsilon_1 n\right)\leq C e^{-\gamma n}
\vspace{0.1cm} \\ \hspace{3cm}
\displaystyle\forall n>-(1+\epsilon_2)\log c/I^*(M,K)\;,
\ena
\eeq
for all $m=1, 2, ..., M$.
\vspace{0.2cm}
\end{lemma}
%
\begin{proof}
Note that
\beq\label{eq:pl_DR_tau1_policy1_1}
\bea{l}
\displaystyle\mathbf{P}_m\left(\DR(\tau_2)>\epsilon_1 n\right) \vspace{0.2cm} \\
\leq\displaystyle\mathbf{P}_m\left(\tau_2>n\right) 
+\displaystyle\mathbf{P}_m\left(\DR(\tau_2)>\epsilon_1 n, \tau_2\leq n\right) \vspace{0.2cm} \\
\ena
\eeq
Since $\tau_2=\tau_1+n_2$, applying Lemmas~\ref{lemma:tau_1_policy1},~\ref{lemma:tau_2_policy1} implies that the first term on the RHS of (\ref{eq:pl_DR_tau1_policy1_1}) decreases exponentially with $n$ for all $n>-(1+\epsilon_2)\log c/I^*(M,K)$ for every fixed $\epsilon_2>0$. It remains to show that the second term on the RHS of (\ref{eq:pl_DR_tau1_policy1_1}) decreases exponentially with $n$.
Let $\bar{j}=\arg\;\max_{j\neq m}S_j(\tau_2), \underline{j}=\arg\;\min_{j\neq m}S_j(\tau_2)$. Let $t_0$ be the smallest integer such that $S_{\underline{j}}(t)\leq S_{\bar{j}}(t)$ for all $t_0<t\leq\tau_2$. As a result, $\DR(\tau_2)>\epsilon_1 n$ implies
\begin{center}
$\displaystyle\sum_{t=t_0}^{\tau_2}\ell_{\bar{j}}\mathbf{1}_{\bar{j}}(t)
                        -\ell_{\underline{j}}\mathbf{1}_{\underline{j}}(t)>\epsilon_1 n$\;.
\end{center}
Note that the second term on the RHS of (\ref{eq:pl_DR_tau1_policy1_1}) can be rewritten as:
\beq\label{eq:pl_DR_tau1_policy1_2}
\bea{l}
\displaystyle\mathbf{P}_m\left(\DR(\tau_2)>\epsilon_1 n, \tau_2\leq n\right) \vspace{0.2cm} \\
=\displaystyle\mathbf{P}_m\left(\DR(\tau_2)>\epsilon_1 n, \tau_2\leq n, t_0\geq\tau_1\right) \vspace{0.2cm} \\
\displaystyle+\mathbf{P}_m\left(\DR(\tau_2)>\epsilon_1 n, \tau_2\leq n, t_0<\tau_1\right)
\ena
\eeq
Let $\underline{N}=\sum_{t=t_0}^{\tau_2}\mathbf{1}_{\underline{j}}(t), \overline{N}=\sum_{t=t_0}^{\tau_2}\mathbf{1}_{\bar{j}}(t)$.

First, we upper bound the first term on the RHS of (\ref{eq:pl_DR_tau1_policy1_2}). Note that for all $\tau_1\leq t_0<t\leq\tau_2$, if $\mathbf{1}_{\underline{j}}(t)=1$ then $\mathbf{1}_{\bar{j}}(t)=1$ (since $S_m(t)>S_j(t)$ for all $t\geq\tau_1$ and the decision maker observes either the $K$ cells with the top $K$ highest sum LLRs or those with the second to the $(K+1)^{th}$ highest sum LLRs). Hence, $\underline{N}\leq\overline{N}$.
Thus,
\beq
\label{eq:pl_DR_tau1_policy1_2_half}
\bea{l}
\displaystyle\displaystyle\sum_{t=t_0}^{\tau_2}\ell_{\bar{j}}\mathbf{1}_{\bar{j}}(t)
                        -\ell_{\underline{j}}\mathbf{1}_{\underline{j}}(t)
                    \vspace{0.2cm} \\
\displaystyle=\sum_{t=t_0}^{\tau_2}\left[\tilde{\ell}_{\bar{j}}\mathbf{1}_{\bar{j}}(t)
                        -\tilde{\ell}_{\underline{j}}\mathbf{1}_{\underline{j}}(t)\right]
                        -D(f||g)\left(\overline{N}-\underline{N}\right)
                    \vspace{0.2cm} \\
\displaystyle\leq\sum_{t=t_0}^{\tau_2}\tilde{\ell}_{\bar{j}}\mathbf{1}_{\bar{j}}(t)
                        -\tilde{\ell}_{\underline{j}}\mathbf{1}_{\underline{j}}(t)
%
\ena
\eeq
Similar to (\ref{eq:lemma_W_star_bound}), applying the Chernoff bound completes the proof for this case.

Next, we upper bound the second term on the RHS of (\ref{eq:pl_DR_tau1_policy1_2}).
Let
$\epsilon_3\triangleq\frac{\epsilon_1}{4D(f||g)}>0$.
Note that
\beq\label{eq:pl_DR_tau1_policy1_3}
\bea{l}
\displaystyle\mathbf{P}_m\left(\DR(\tau_2)>\epsilon_1 n, \tau_2\leq n, t_0<\tau_1\right) \vspace{0.2cm} \\
\leq\displaystyle\mathbf{P}_m\left(\tau_1>\epsilon_3 n\right) \vspace{0.2cm} \\
\displaystyle+\mathbf{P}_m\left(\DR(\tau_2)>\epsilon_1 n, \tau_2\leq n, t_0<\tau_1, \tau_1\leq\epsilon_3 n\right) \;.
\ena
\eeq
The first term on the RHS of (\ref{eq:pl_DR_tau1_policy1_3}) decreases exponentially with $n$ by Lemma~\ref{lemma:tau_1_policy1}. Thus, it remains to show that the second term on the RHS of (\ref{eq:pl_DR_tau1_policy1_3}) decreases exponentially with $n$.
Note that $\DR(\tau_2)>\epsilon_1 n$ implies
$\bea{l}
\displaystyle\left(\sum_{t=t_0}^{\tau_1}\ell_{\bar{j}}\mathbf{1}_{\bar{j}}(t)
                        -\ell_{\underline{j}}\mathbf{1}_{\underline{j}}(t)\right)
+\left(\sum_{t=\tau_1+1}^{\tau_2}\ell_{\bar{j}}\mathbf{1}_{\bar{j}}(t)
                        -\ell_{\underline{j}}\mathbf{1}_{\underline{j}}(t)\right)  \vspace{0.2cm} \\ \hspace{5cm}
                       \displaystyle >\epsilon_1 n \;. \vspace{0.2cm}
\ena$
Therefore, the second term on the RHS of (\ref{eq:pl_DR_tau1_policy1_3}) can be rewritten as:
\beq\label{eq:pl_DR_tau1_policy1_4}
\bea{l}
\displaystyle\mathbf{P}_m\left(\DR(\tau_2)>\epsilon_1 n, \tau_2\leq n, t_0<\tau_1, \tau_1\leq\epsilon_3 n\right) \vspace{0.2cm} \\
\displaystyle\leq\mathbf{P}_m\left(\sum_{t=t_0}^{\tau_1}\ell_{\bar{j}}\mathbf{1}_{\bar{j}}(t)
-\ell_{\underline{j}}\mathbf{1}_{\underline{j}}(t)
>\frac{\epsilon_1 n}{2}, \right. \\ \vspace{0.2cm} \hspace{3cm}\left.
\tau_2\leq n, t_0<\tau_1, \tau_1\leq\epsilon_3 n\right) \vspace{0.2cm} \\
\displaystyle+\mathbf{P}_m\left(\sum_{t=\tau_1+1}^{\tau_2}\ell_{\bar{j}}\mathbf{1}_{\bar{j}}(t)
-\ell_{\underline{j}}\mathbf{1}_{\underline{j}}(t)
>\frac{\epsilon_1 n}{2}, \right. \\ \vspace{0.2cm} \hspace{3cm}\left.
\tau_2\leq n, t_0<\tau_1, \tau_1\leq\epsilon_3 n\right)
\ena
\eeq
Note that for all $t_0<\tau_1+1\leq t\leq\tau_2$, if $\mathbf{1}_{\underline{j}}(t)=1$ then $\mathbf{1}_{\bar{j}}(t)=1$. As a result, cell $\bar{j}$ is probed more frequently between $\tau_1+1$ and $\tau_2$. Thus, the second term on the RHS of (\ref{eq:pl_DR_tau1_policy1_4}) decreases exponentially with $n$ using a similar argument as in (\ref{eq:pl_DR_tau1_policy1_2_half}).
Next, it remains to show that the first term on the RHS of (\ref{eq:pl_DR_tau1_policy1_4}) decreases exponentially with $n$.
Note that
\beq
\bea{l}
\displaystyle\displaystyle\sum_{t=t_0}^{\tau_1}\ell_{\bar{j}}\mathbf{1}_{\bar{j}}(t)
                        -\ell_{\underline{j}}\mathbf{1}_{\underline{j}}(t)
                    \vspace{0.2cm} \\
\displaystyle\leq\sum_{t=t_0}^{\tau_1}\left[\tilde{\ell}_{\bar{j}}\mathbf{1}_{\bar{j}}(t)
                        -\tilde{\ell}_{\underline{j}}\mathbf{1}_{\underline{j}}(t)\right]
                        +D(f||g)\tau_1
                    \vspace{0.2cm} \\
\displaystyle\leq\sum_{t=t_0}^{\tau_1}\left[\tilde{\ell}_{\bar{j}}\mathbf{1}_{\bar{j}}(t)
                        -\tilde{\ell}_{\underline{j}}\mathbf{1}_{\underline{j}}(t)\right]
                        +\frac{\epsilon_1}{4}n
                    \vspace{0.2cm} \\
\ena
\eeq
for all $\tau_1\leq\epsilon_3 n$.\\
As a result,
\beq
\bea{l}
\displaystyle\displaystyle\sum_{t=t_0}^{\tau_1}\ell_{\bar{j}}\mathbf{1}_{\bar{j}}(t)
                        -\ell_{\underline{j}}\mathbf{1}_{\underline{j}}(t) >\frac{\epsilon_1}{2} n
\ena
\eeq
implies
\beq
\bea{l}
\displaystyle\sum_{t=t_0}^{\tau_1}\left[\tilde{\ell}_{\bar{j}}\mathbf{1}_{\bar{j}}(t)
                        -\tilde{\ell}_{\underline{j}}\mathbf{1}_{\underline{j}}(t)\right]
                        >\frac{\epsilon_1}{4}n
                    \vspace{0.2cm} \\
\ena
\eeq
for all $\tau_1\leq\epsilon_3 n$.\\
Similar to (\ref{eq:lemma_W_star_bound}), applying the Chernoff bound completes the proof.
\vspace{0.2cm}
\end{proof}
\begin{definition}
The dynamic range $\DR_j(t)$ for $j\neq m$ at time~$t$ is defined as follows:
\beq
\displaystyle \DR_j(t)\triangleq\max_{k\neq m}S_k(t)-S_j(t) \;. \vspace{0.2cm}
\eeq
\end{definition}
Let
\beq
\label{eq:selection_policy1}\bea{l}
\displaystyle \eta\triangleq \vspace{0.2cm}\\ \hspace{0.0cm}
\begin{cases} \frac{D(f||g)}{I^*(M,K)}\log c \;\;, \;\;
                                                \mbox{if \;$K=M$}   \vspace{0.3cm}\\
              \log c \;\;, \;\;
              \mbox{if \;$K<M$ and $I^*(M,K)=K\frac{D(f||g)}{M-1}$}   \vspace{0.3cm}\\
              \frac{D'(f||g)}{I^*(M,K)}\log c \;\;, \vspace{0.1cm}\\ \hspace{0.3cm}
              \mbox{if \;$K<M$ and $I^*(M,K)=D(g||f)+\frac{(K-1)D(f||g)}{M-1}$}
\end{cases}
\ena
\eeq
where $D'(f||g)$ is defined in (\ref{eq:D_prime_fg}).
\begin{definition}
$\tau_3^j$ denotes the smallest integer such that $\sum_{i=\tau_1+1}^{n}\ell_j(i)\mathbf{1}_j(i)\leq\eta+\DR_j(\tau_1)$ for $j\neq m$ for all $n\geq\tau_3^j\geq\tau_2$. We also define $\tau_3\triangleq\max_{j\neq m}\tau_3^j$. \vspace{0.2cm}
\end{definition}
Note that $\tau_3^j\geq\tau_2$ by definition (i.e., both $\tau_2$ has passed and the inequality holds for all $n\geq\tau_3^j$). \vspace{0.2cm}
\begin{remark}
Using some algebraic manipulations, it can be verified that $\Delta S_{m,j}(n)\geq-\log c$ for all $j\neq m$ for all $n\geq\tau_3^j$. Since $\tau_3=\max_{j\neq m}\tau_3^j$ we have $\Delta S(n)=S_m(n)-S_{m^{(2)}(n)}(n)\geq-\log c$ for all $n\geq\tau_3$. It should be noted that $\tau_3$ depends on the future and is not a stopping time. The decision maker does not know whether it has arrived. However, it is used to upper bound the actual stopping time under DGF. Note that the stopping time under DGF stops the sampling once $\Delta S(n)\geq-\log c$ \emph{first} occurs. If $\Delta S(n)\geq-\log c$ first occurs once $\tau_3$ occurs, then $\tau=\tau_3$. Otherwise, $\tau<\tau_3$. \vspace{0.2cm}
\end{remark}
\begin{definition}
$n_3\triangleq\tau_3-\tau_2$ denotes the total amount of time between $\tau_2$ and $\tau_3$. \vspace{0.2cm}
\end{definition}
%
\begin{lemma}
\label{lemma:tau_3_policy1}
Assume that the DGF policy is implemented indefinitely.
Then, for every fixed $\epsilon>0$ there exist $C>0$ and $\gamma>0$ such that
\beq\label{eq:l_tau3}
\mathbf{P}_m\left(n_3>n\right)\leq C e^{-\gamma n} \;\;\;\; \forall n>-\epsilon\log c/I^*(M,K)\;,
\eeq
for all $m=1, 2, ..., M$.
\vspace{0.2cm}
\end{lemma}
%
\begin{proof}
Let $N_3^j$ for $j\neq m$ denote the total number of observations, taken from cell $j$ between $\tau_2$ and $\tau_3^j$.
Note that $n_3\leq\sum_{j\neq m}N_3^j$. Thus, it suffices to show that $\mathbf{P}_m\left(N_3^j>n\right)$ decreases exponentially with $n$.
Note that
\beq\label{eq:pl_tau3_1}
\bea{l}
\mathbf{P}_m\left(N_3^j>n\right) 
\leq\mathbf{P}_m\left(\DR(\tau_2)>n\frac{D(f||g)}{2}\right) \vspace{0.2cm} \\ \hspace{2.2cm}
+\mathbf{P}_m\left(N_3^j>n \; | \; \DR(\tau_2)\leq n\frac{D(f||g)}{2}\right)
\;.
\ena
\eeq
By Lemma~\ref{lemma:DR_tau1_policy1}, the first term on the RHS of (\ref{eq:pl_tau3_1}) decreases exponentially with $n$ for all $n>-\epsilon\log c/I^*(M,K)$. Thus, it remains to show that the second term on the RHS of (\ref{eq:pl_tau3_1}) decreases exponentially with $n$. \\
Let $t_1, t_2, ...$ denote the time indices when cell $j$ is observed between $\tau_2$ and $\tau_3^j$. Since $\tau_2$ has occurred and $\DR(\tau_2)\leq n\frac{D(f||g)}{2}$, $\tau_3^j$ occurs once $\sum_{i=1}^{r}-\ell_{j}(t_i)\geq n\frac{D(f||g)}{2}$ holds for all $r\geq N_3^j$.
As a result,
\beq
\bea{l}
\displaystyle\mathbf{P}_m\left(N_3^j>n \; | \; \DR(\tau_2)\leq n\frac{D(f||g)}{2}\right)\vspace{0.1cm} \\ %
\displaystyle\leq\mathbf{P}_m\left(\inf_{r>n}\sum_{i=1}^{r}-\ell_j(t_i)<n\frac{D(f||g)}{2}\right)
\vspace{0.1cm} \vspace{0.1cm} \\
\displaystyle\leq\sum_{r=n}^{\infty}\mathbf{P}_m\left(\sum_{i=1}^{r}\tilde{\ell}_j(t_i)>r\frac{D(f||g)}{2}\right)
\vspace{0.1cm} \vspace{0.1cm}\;.
\ena
\eeq
Thus, it suffices to show that there exists $\gamma>0$ such that $\mathbf{P}_m\left(\sum_{i=1}^{n}\tilde{\ell}_j(t_i)>n\frac{D(f||g)}{2}\right)\leq e^{-\gamma n}$. Applying the Chernoff bound and using the i.i.d. property of $\tilde{\ell}_j(t_i)$ completes the proof.
\vspace{0.2cm}
\end{proof}

\begin{lemma}
\label{lemma:expected_time_policy1}
The expected detection time $\tau$ under the DGF policy is upper bounded by:
\beq\label{eq:lemma_expected_time}
\mathbf{E}_m(\tau)\leq -\left(1+o(1)\right)\frac{\log(c)}{I^*(M,K)} \;,
\eeq
for $m=1, ..., M$.
\vspace{0.2cm}
\end{lemma}
%
\begin{proof}
Note that $\tau\leq\tau_3=\tau_1+n_2+n_3$. Thus, combining Lemmas~\ref{lemma:tau_1_policy1},~\ref{lemma:tau_2_policy1} and~\ref{lemma:tau_3_policy1} completes the proof.
\vspace{0.2cm}
\end{proof}

Combining Lemmas~\ref{lemma:error_policy1},~\ref{lemma:expected_time_policy1} and Theorem~\ref{th:lower_bound} yields the asymptotic optimality property of the DGF policy, presented in Theorem~\ref{th:optimality_policy1}.

\subsection{Proof of Theorem~\ref{th:optimality_Chernoff}}
\label{app:proof_Chernoff}

Following the same argument as in~\cite{Nitinawarat_2013_Controlled}, it suffices to show that $\mathbf{P}_m\left(\tau_1>n\right)$ decreases polynomially with $n$ to prove the theorem. Since $D(g||f)>0$, the KL divergence between the true hypothesis $m$ and any false hypothesis $j\neq m$ is strictly positive under any observed cell. For the ease of presentation, consider the case where $K=1$. In the case where $D(g||f)\geq D(f||g)/(M-1)$, the Chernoff test selects $m^{(1)}(n)$ for all $n$. As a result, exponential decay of $\mathbf{P}_m\left(\tau_1>n\right)$ follows directly from Lemma~\ref{lemma:tau_1_policy1}. In the case where $D(g||f)<D(f||g)/(M-1)$, the Chernoff test selects $m^{j}(n)$ for $j\neq 1$ randomly for all $n$. As a result, polynomial decay of $\mathbf{P}_m\left(\tau_1>n\right)$ follows by a similar argument as in~\cite{Nitinawarat_2013_Controlled} for the extended Chernoff test. Note that the proof directly applies to the case where $K>1$ since $P_m(\tau_1>n)$ decreases as the number $K$ of observations collected at a time increases.

\newcommand*{\QEDA}{\hfill\ensuremath{\blacksquare}}
\QEDA

\subsection{Proof of Theorem~\ref{th:optimality_policy2}}
\label{app:proof_policy2}

The proof follows a similar line of arguments as in the proof of Theorem~\ref{th:optimality_policy1}. Hence, we provide here only a sketch of the proof. First, similar to Lemma~\ref{lemma:error_policy1}, it can be verified that declaring the target locations once $S_{m^{(L)}}-S_{m^{(L+1)}}\geq-\log c$ occurs achieves an error probability $O(c)$. Second, similar to Lemma~\ref{lemma:expected_time_policy1}, it can be verified that the detection time approaches $-\log c/I^*(M,K)$. For example, if $\frac{D(g||f)}{L}\geq\frac{D(f||g)}{M-L}$ and $K\geq L$ then all the $L$ targets and a fraction $r=\frac{K-L}{M-L}$ of the false hypotheses are observed at each given time in the asymptotic regime. Therefore, the detection time approaches $\frac{-\log c}{D(g||f)+rD(f||g))}$. Similar arguments apply to the rest of the cases.

\QEDA

\subsection{Proof of Theorem~\ref{th:optimality_policy3}}
\label{app:proof_policy3}

The proof follows a similar line of arguments as in the proofs of Theorems~\ref{th:optimality_policy1} and~\ref{th:optimality_policy2}. Hence, we provide only a sketch of the proof. With minor modifications to Theorem~\ref{th:optimality_policy2}, it can be verified that (\ref{eq:asymptotic_performance_policy3}) is the asymptotic lower bound on the Bayes risk when the number of targets $\ell$ is known, $K=1$ and (\ref{eq:M}) holds. Similar to Lemma~\ref{lemma:error_policy1}, it can be verified that declaring a target once $S_m(n)\geq -\log c$ occurs achieves an error probability $O(c)$. Following a similar argument as in Lemma~\ref{lemma:tau_1_policy1}, it can be verified that the $\ell$ targets are tested before testing the $M-\ell$ normal processes with high probability in the asymptotic regime. Since the decision maker declares the target locations once $S_m(n)\geq -\log c$ for any $m$, the Bayes risk approaches (\ref{eq:asymptotic_performance_policy3}) as $c\rightarrow 0$.

\QEDA

\bibliographystyle{ieeetr}

%
%
\end{document}

%% file: macros.tex
\newcommand{\E}{\ensuremath{\hbox{\textbf{E}}}}
\newcommand{\DR}{\ensuremath{\hbox{DR}}}

\newtheorem{theorem}{Theorem}
\newtheorem{lemma}{Lemma}
\newtheorem{remark}{Remark}
\newtheorem{definition}{Definition}

\newcommand{\beq}{\begin{equation}}
\newcommand{\eeq}{\end{equation}}
\newcommand{\bea}{\begin{array}}
\newcommand{\ena}{\end{array}}
\newcommand{\bds}{\begin {itemize}}
\newcommand{\eds}{\end {itemize}}
\newcommand{\bdf}{\begin{definition}}
\newcommand{\blm}{\begin{lemma}}
\newcommand{\edf}{\end{definition}}
\newcommand{\elm}{\end{lemma}}
\newcommand{\bthm}{\begin{theorem}}
\newcommand{\ethm}{\end{theorem}}
\newcommand{\bprp}{\begin{prop}}
\newcommand{\eprp}{\end{prop}}
\newcommand{\bcl}{\begin{claim}}
\newcommand{\ecl}{\end{claim}}
\newcommand{\bcr}{\begin{coro}}
\newcommand{\ecr}{\end{coro}}
\newcommand{\bquest}{\begin{question}}
\newcommand{\equest}{\end{question}}

\newcommand{\larrow}{{\larrow}}

\newcommand{\nin}{{\not \in}}

\def\urltilda{\kern -.15em\lower .7ex\hbox{\~{}}\kern .04em}